\documentclass[twocolumn,notitlepage,floatfix,reprint,superscriptaddress,aps]{revtex4-2}

\usepackage{graphicx}
\usepackage{amsmath}
\usepackage{amsthm}
\usepackage{amssymb}
\usepackage{latexsym}
\usepackage{array}
\usepackage{hyperref}
\usepackage{amsfonts}
\usepackage{dsfont}
\usepackage{mathrsfs}
\usepackage{verbatim}
\usepackage{bbold}
\usepackage[normalem]{ulem}
\usepackage{upgreek}
\usepackage{makecell}
\usepackage{adjustbox}
\usepackage{algorithm}
\usepackage{algpseudocode}
\usepackage{color}
\usepackage{bm}
\usepackage{times}
\usepackage{booktabs}
\usepackage{multirow}

\usepackage{textcomp}
\usepackage[dvipsnames]{xcolor}

\usepackage{bm}
\usepackage{times}

\newcommand{\abs}[1]{\left|#1\right|}
\newcommand{\norm}[1]{\left\|#1\right\|}
\newcommand{\ket}[1]{\left|#1\right\rangle}
\newcommand{\bra}[1]{\left\langle #1\right|}

\newcommand{\ketbra}[2]{\ket{#1}\!\!\bra{#2}}
\newcommand{\Tr}{\operatorname{Tr}}
\newcommand{\Var}{\operatorname{Var}}

\DeclareMathOperator{\supp}{supp}
\newcommand{\id}{\mathds{1}}
\newcommand{\hU}{\hat U}
\newcommand{\hrho}{\hat\rho}
\newcommand{\EP}{\mathrm{EP}}
\newcommand{\EPD}{\mathrm{EPD}}
\newcommand{\LC}{\mathrm{LC}}
\newcommand{\Haar}{\mathrm{Haar}}
\newcommand{\prodens}{\mathrm{prod}}
\newcommand{\SM}{Supplemental Material}

\newcommand{\hid}{\hat{\id}}
\newcommand{\hsigma}{\hat{\sigma}}
\newcommand{\hG}{\hat{G}}
\newcommand{\hH}{\hat{H}}
\newcommand{\hO}{\hat{O}}
\newcommand{\hB}{\hat{B}}
\newcommand{\hM}{\hat{M}}
\newcommand{\hK}{\hat{K}}
\newcommand{\hT}{\hat{T}}
\newcommand{\hS}{\hat{S}}

\newcommand{\hE}{\hat{E}}
\newcommand{\hA}{\hat{A}}
\newcommand{\hR}{\hat{R}}
\newcommand{\hX}{\hat{X}}
\newcommand{\hY}{\hat{Y}}
\newcommand{\hZ}{\hat{Z}}

\theoremstyle{plain}
\newtheorem{theorem}{Theorem}
\newtheorem{proposition}{Proposition}
\newtheorem{lemma}{Lemma}
\newtheorem{corollary}{Corollary}
\theoremstyle{definition}
\newtheorem{definition}{Definition}
\newtheorem{remark}{Remark}

\makeatletter
\renewcommand\part[1]{%
  \clearpage
  \onecolumngrid
  \section*{#1}
  
}
\makeatother

\makeatletter
\let\origaddcontentsline\addcontentsline
\renewcommand{\addcontentsline}[3]{}
\makeatother

\begin{document}

\part{}

\title{Breaking the One-Dimensional Expressibility--Trainability Tradeoff}

\author{Kyoungho~Cho}
\affiliation{Institute for Convergence Research and Education in Advanced Technology, Yonsei University, Seoul 03722, Republic of Korea}
\affiliation{Department of Statistics and Data Science, Yonsei University, Seoul 03722, Republic of Korea}

\author{Yu-Seong~Jeon}
\affiliation{Department of Physics, Hanyang University, Seoul 04763, Republic of Korea}

\author{Jinhyoung~Lee}\email{hyoung@hanyang.ac.kr}
\affiliation{Department of Physics, Hanyang University, Seoul 04763, Republic of Korea}
\affiliation{Hanyang Institute for Quantum Science and Quantum Technology, Hanyang University, Seoul 04763, Republic of Korea}

\author{Jeongho~Bang}\email{jbang@yonsei.ac.kr}
\affiliation{Institute for Convergence Research and Education in Advanced Technology, Yonsei University, Seoul 03722, Republic of Korea}
\affiliation{Department of Quantum Information, Yonsei University, Incheon 21983, Republic of Korea}

\date{\today}

\begin{abstract}
Expressive parameterized quantum circuits (PQCs) are often designed under a dilemma: the growth of expressibility and entangling power (EP) that improves Hilbert-space coverage is also expected to randomize an ansatz and activate barren-plateau (BP) conditions. We show that this dilemma is not a one-dimensional tradeoff. The usual picture collapses three inequivalent objects---parameter-ensemble coverage, fixed-circuit entangling response, and local gradient moments---into one scalar narrative. For a fixed circuit probed by Haar-product inputs, EP is a global two-copy mean of the output-entanglement distribution, whereas entangling-power deviation (EPD) is a global four-copy fluctuation descriptor. Gradient variance, however, is a local two-copy contraction selected by a parameter light cone and a cost observable. This moment hierarchy yields an analytic separation: equal EP need not imply equal trainability, as witnessed by equal-EP circuits with different EPDs and different gradient variances. These separations turn EP and EPD into a two-dial design rule for PQC ansatzes: EP measures how far the circuit has moved along the coverage dial, while EPD monitors whether input-dependent variability remains. We find that ansatz routes can reach high, Haar-like coverage before EPD and gradient variance collapse, showing that coverage and BP activation are distinct crossover events. The EP/EPD framework thus breaks the apparent one-dimensional expressibility–trainability tradeoff into a practical design rule: search for highly expressive PQCs in the window where coverage is high but BP-like homogenization has not yet erased trainable structure.
\end{abstract}

\maketitle

{\em Introduction.}---Parameterized quantum circuits (PQCs) are judged by two requirements that pull in opposite directions. PQCs must generate sufficiently rich state families for variational quantum algorithms and/or quantum machine learning~\cite{Peruzzo2014VQE,Farhi2014QAOA,Kandala2017HEA,Havlicek2019QML,Cerezo2021VQA}. PQCs must also retain gradients large enough to be estimated and used by a classical optimizer. The first requirement is commonly quantified through expressibility, namely the closeness of parameter-sampled states to Haar coverage, together with the average Meyer-Wallach entanglement generated by the circuit, which we call entangling power (EP)~\cite{Sim2019expressibility,Zanardi2000EP}. The second requirement is trainability: the circuit must not become so randomized that the cost landscape loses usable gradient information. Barren-plateau theory made this requirement sharp by proving that sufficiently random circuits can have exponentially small gradients~\cite{McClean2018barrenplateaus,Cerezo2021BP,Holmes2022expressibility}. These facts are often compressed into the slogan that increasing expressibility together with EP drives a circuit toward Haar-random structure and hence toward barren-plateau (BP) behavior.

This slogan is useful as a warning against blindly random circuits, but it is too coarse to serve as a design principle. A BP criterion tests whether the gradient-relevant moment has already randomized; it does not say what kind of coverage was gained, nor whether input-dependent structure remains before that randomization. The reason is threefold. First, expressibility is a property of a parameter ensemble, whereas EP is naturally a property of one fixed unitary tested on many product inputs. Second, EP is a global scalar, whereas the gradient of a particular parameter is a local object selected by a cost observable and an effective light cone. Third, a mean does not determine a fluctuation profile. These distinctions help explain why the descriptor-performance relations in classification learning studies can be nuanced: expressibility may correlate with learning accuracy, while the average EP can be much less predictive~\cite{Hubregtsen2021evaluation}.

In this Letter, we formulate the missing structure as a moment hierarchy. For that, we extend the entangling-power deviation (EPD) viewpoint, previously developed for bipartite gates~\cite{ChoBang2026EPD}, to $n$-qubit PQCs. The result is a two-dial design principle: First, a coverage dial records how far the ansatz has moved toward high mean entangling strength or Haar-like coverage. Second, a variability dial records whether the output entanglement generated from product inputs remains input-dependent. We show analytically that these dials are distinct, prove that equal mean EP does not imply equal trainability, and show that useful ansatz routes can reach high coverage before EPD and gradient variance collapse.

Our central claim is that the aforementioned slogan, i.e., the expressibility--vs--trainability tradeoff, is not a hard one-dimensional obstruction: high coverage can be reached before the BP condition is locally activated. In this window, EP records that the ansatz has become sufficiently nonlocal, while EPD records that product-input responses have not homogenized. Thus, our EP/EPD framework identifies high-coverage, BP-free candidates by separating mean coverage from the collapse of input-dependent variability. It provides a constructive screening principle for PQC design: maximize coverage without crossing into the low-EPD regime where gradient variance is expected to collapse.

{\em Fixed-circuit descriptors.}---We start by separating the two ensemble levels that are often conflated.

\begin{definition}[Ensemble levels]
Let $\hU_L(\bm{\theta})$ be a depth-$L$ PQC on $\mathcal{H}_n=(\mathbb{C}^2)^{\otimes n}$. The parameter ensemble $\nu_L$ samples circuit instances $\hU_L(\bm{\theta})$ and underlies expressibility. The product-input ensemble
\begin{eqnarray}
\mu_{\prodens} = \bigotimes_{j=1}^n\mu_{\Haar}^{(1)},
\quad
\ket{\psi} = \bigotimes_{j=1}^n\ket{\psi_j},
\label{eq:product_ensemble_main}
\end{eqnarray}
with independent one-qubit Haar states $\ket{\psi_j}$, probes a single fixed circuit. These two ensembles are not interchangeable.
\end{definition}

For a pure state $\ket{\phi}$, we use the Meyer-Wallach entanglement~\cite{MeyerWallach2002}
\begin{eqnarray}
Q(\ket{\phi}) = \frac{2}{n}\sum_{j=1}^n\left(1-\Tr\hat{\rho}_j^2\right),
\quad
\hat{\rho}_j = \Tr_{\bar j}\ketbra{\phi}{\phi}.
\label{eq:Q_def_main}
\end{eqnarray}
It satisfies $0\le Q\le1$ and vanishes on fully product states. 

We now define the two fixed-circuit quantities used throughout this Letter.

\begin{definition}[Fixed-circuit EP and EPD]
For a fixed unitary $\hU \in U(2^n)$,
\begin{eqnarray}
\EP(\hU) &=& \mathbb E_{\psi\sim\mu_{\prodens}}\bigl[Q(\hU\ket{\psi})\bigr], \nonumber \\
\EPD(\hU) &=& \sqrt{\Var_{\psi\sim\mu_{\prodens}}\bigl[Q(\hU\ket{\psi})\bigr]} .
\label{eq:EP_EPD_def_main}
\end{eqnarray}
EP is the mean entangling strength of a circuit instance. EPD is the standard deviation of the same output-entanglement distribution, and hence measures the input-state sensitivity~\cite{ChoBang2026EPD}. Both are invariant under local pre- and post-unitaries.
\end{definition}

The fixed-circuit viewpoint has two practical advantages. First, it separates the nonlocal action of a circuit block from the arbitrary local basis in which the block is written. Second, it lets us ask a question: does the same circuit entangle all product inputs in nearly the same way, or does it preserve strong input dependence? This second question is invisible to the mean but is directly relevant to the trainability, because the gradients are also response functions of a circuit with respect to perturbations of an input, a parameter, or an observable.

The distinction between the two descriptors becomes exact in copy space. Let $S_j$ swap $j$-qubit between two copies and let us define $\hat{K}_Q = \frac{2}{n} \sum_{j=1}^n (\hat{\id} - \hat{S}_j)$. Then, $Q(\ket{\phi})=\Tr[\hat{K}_Q \hat{\rho}^{\otimes2}]$. On four copies, we define $\hat{T}_Q = \hat{K}_Q^{(12)}\hat{K}_Q^{(34)}$, where the superscripts denote copy pairs. For the fixed-circuit moment
\begin{eqnarray}
\hat{M}_{\hU}^{(t)}=\mathbb E_{\psi \sim \mu_{\prodens}} \left[\left(\hU\ketbra{\psi}{\psi}\hU^\dagger\right)^{\otimes t}\right],\label{eq:MU_t_main}
\end{eqnarray}
we have the following structural identity.

\begin{theorem}[Moment hierarchy]
\label{thm:1}
For every fixed circuit $\hU$,
\begin{eqnarray}
\EP(\hU) &=& \Tr\bigl[\hat{K}_Q \hat{M}_{\hU}^{(2)}\bigr], \nonumber \\
\EPD(\hU)^2 &=& \Tr\bigl[\hat{T}_Q \hat{M}_{\hU}^{(4)}\bigr] - \EP(\hU)^2.
\label{eq:main_hierarchy_theorem}
\end{eqnarray}
Thus, EP is a global two-copy scalar, while EPD is a global four-copy fluctuation descriptor.
\end{theorem}

\begin{proof}[Proof sketch]
The local swap identity $\Tr[\hat{S}_j (\hat{A} \otimes \hat{B})] = \Tr[(\Tr_{\bar j} \hat{A}) (\Tr_{\bar j} \hat{B})]$ gives Eq.~(\ref{eq:Q_def_main}) as $\Tr[\hat{K}_Q \hat{\rho}^{\otimes 2}]$. Squaring the same two-copy expectation on the two disjoint copy pairs $(1,2)$ and $(3,4)$ gives $Q(\ket{\phi})^2 = \Tr[\hat{T}_Q \hat{\rho}^{\otimes4}]$. Averaging these two identities over Haar-product inputs after applying $\hU$ yields Eq.~(\ref{eq:main_hierarchy_theorem}). The full swap derivation and norm bounds are given in Sec.~II of the \SM.
\end{proof}

\begin{proposition}[Spectral obstruction]
\label{prop:1}
For every $t \ge 1$, $\hat{M}_{\hU}^{(t)} = \hU^{\otimes t} \hat{M}_{\prodens}^{(t)} \hU^{\dagger\otimes t}$. Hence, the spectrum of $\hat{M}_{\hU}^{(t)}$ is independent of $\hU$. For $n > 1$ and $t \ge 2$, it differs from the spectrum of the full Haar pure-state moment.
\end{proposition}

\begin{proof}[Proof sketch]
Linearity pulls $\hU^{\otimes t}$ outside the product-input expectation. The product-input moment is proportional to the projector onto the locally symmetric subspace $\bigotimes_j\mathrm{Sym}^t(\mathbb C^2_j)$, of dimension $(t+1)^n$. The full Haar moment is proportional to the projector onto $\mathrm{Sym}^t(\mathcal H_n)$, of dimension $\binom{2^n+t-1}{t}$. These subspaces are not equal for $n > 1, t \ge 2$; the detailed support comparison is given in Sec.~I of the \SM. Therefore, Haar is used below as a scalar entanglement benchmark and as a local light-cone benchmark, not as a global operator limit of fixed-circuit product-input moments.
\end{proof}

This obstruction is important for interpretation. When a deep parameterized circuit is called ``Haar-like'', the statement usually concerns a particular observable, fidelity distribution, or low-order local contraction. It cannot mean that one fixed unitary acting on Haar-product inputs literally produces the full Haar moment operator on the global Hilbert space.


{\em Gradient variance is local.}---Consider a cost
\begin{eqnarray}
C(\bm{\theta},\psi) = \Tr[\hat{O}\hat{\rho}_\psi(\bm{\theta})],
\end{eqnarray}
where $\hat{\rho}_\psi(\bm{\theta})=\hU(\bm{\theta})\ketbra{\psi}{\psi}\hU(\bm{\theta})^\dagger$ and $\hat{O} = \hat{O}^\dagger$ is the bounded cost observable whose expectation value defines the optimization objective. Here and below, $\bm{\theta}$ denotes the full parameter vector, whereas $\theta_i$ denotes its $i$th component. Suppose that $\theta_i$ occurs in the elementary gate $\exp(-i \theta_i \hat{G}_i)$ with Hermitian generator $\hat{G}_i$. We decompose the circuit as
\begin{eqnarray}
\hU(\bm{\theta})=\hU_i^{\rm post}\exp(-i \theta_i \hat{G}_i)\hU_i^{\rm pre},
\label{eq:prepost_decomposition_main}
\end{eqnarray}
where $\hU_i^{\rm pre}$ collects all gates before the parameterized gate and $\hU_i^{\rm post}$ collects all gates after it. The output-picture generator is therefore $\hat{H}_i^{\rm out} = \hU_i^{\rm post} \hat{G}_i \hU_i^{{\rm post},\dagger}$, so
\begin{eqnarray}
\partial_{\theta_i}\hat{\rho}_\psi = - i \bigl[ \hat{H}_i^{\rm out}, \hat{\rho}_\psi \bigr].
\label{eq:state_derivative_main}
\end{eqnarray}
Only gates after the parameter, contained in $\hU_i^{\rm post}$, dress the generator; the cost observable then selects the part of that dressed operator that can contribute to the gradient.

\begin{definition}[Effective gradient observable]
For parameter $\theta_i$ and cost observable $\hat{O}$, define
\begin{eqnarray}
\hat{B}_i = i \bigl[ \hat{H}_i^{\rm out}, \hat{O} \bigr], 
\quad 
\LC(i) = \mathrm{supp}(\hat{B}_i),
\label{eq:Bi_def_main}
\end{eqnarray}
where $\hat{B}_i$ is the effective gradient observable; it is Hermitian and traceless because it is $i$ times a commutator of Hermitian operators. The set $\LC(i)$ is its effective light cone. Let $r_i=\abs{\LC(i)}$, $d_i=2^{r_i}$, and let $\hat{\sigma}_\psi^{(i)}=\Tr_{\overline{\LC(i)}}[\hat{\rho}_\psi]$ be the reduced output state on $\LC(i)$. Here, $\mathrm{supp}(\hat{X})$ denotes the smallest set of qubits outside which $\hat{X}$ acts trivially, i.e., $\hat{X}=\hat{X}_{\mathrm{supp}}\otimes\hat{\id}_{\overline{\mathrm{supp}}}$ up to identity factors.
\end{definition}

\begin{proposition}[Local two-copy gradient criterion]
\label{prop:2}
For Haar-product inputs,
\begin{eqnarray}
g_i(\psi) = \partial_{\theta_i}C = \Tr[\hat{B}_i \hat{\sigma}_\psi^{(i)}],
\quad
\mathbb E_\psi[g_i(\psi)]=0,
\end{eqnarray}
and
\begin{eqnarray}
\Var_\psi [g_i(\psi)] &=&\Tr\left[ \left( \hat{B}_i \otimes \hat{B}_i \right) \hat{M}_{i,\LC}^{(2)} \right], \nonumber \\
M_{i,\LC}^{(2)} &=&\mathbb E_\psi \left[ \bigl( \hat{\sigma}_\psi^{(i)} \bigr)^{\otimes 2} \right].
\label{eq:grad_var_local_main}
\end{eqnarray}
The Haar benchmark on the light cone is
\begin{eqnarray}
\Var_{\Haar,r_i}[g_i] = \frac{\Tr\bigl( \hat{B}_i^2 \bigr)}{d_i(d_i+1)} \le \frac{4\|{\hat{G}_i}\|_\infty^2 \|{\hat{O}}\|_\infty^2}{2^{r_i}+1}.
\label{eq:haar_scale_main}
\end{eqnarray}
\end{proposition}

\begin{proof}[Proof sketch]
Eq.~(\ref{eq:state_derivative_main}) gives $g_i=\Tr[\hat{B}_i \hat{\rho}_\psi]$. Since $\hat{B}_i$ is supported only on $\LC(i)$, this is the reduced-state contraction above. Squaring the linear functional gives $g_i^2=\Tr\bigl[(\hat{B}_i \otimes \hat{B}_i)\bigl(\hat{\sigma}_\psi^{(i)}\bigr)^{\otimes2}\bigr]$. Averaging gives Eq.~(\ref{eq:grad_var_local_main}). The mean vanishes because the product-input first moment is maximally mixed and $\Tr\hat{B}_i=0$. The Haar expression follows from the standard second moment $(\hat{\id}+\hat{S})/[d_i(d_i+1)]$ for a traceless observable. Details, including the dressed-generator construction, light-cone reduction, and local moment-gap form, are given in Secs.~III~A--III~C of the \SM.
\end{proof}

Proposition~\ref{prop:1} and Theorem~\ref{thm:1} explain why a one-number global mean descriptor cannot determine trainability. EP is a single global contraction of $\hat{M}^{(2)}_{\hU}$. A gradient variance is a different contraction, local to $\LC(i)$ and depending on $i$ and $\hat{O}$. EPD is not the trainability invariant either, but because it probes fourth-order fluctuation structure, it contains information discarded by EP.

Operationally, Eq.~(\ref{eq:grad_var_local_main}) says that the plateau question must be asked parameter by parameter. If $\hat{O}$ is local and the forward cone of $\hat{G}_i$ does not reach it, then $\hat{B}_i=0$ and the gradient vanishes for a purely causal reason. If the cone is small, the Haar scale in Eq.~(\ref{eq:haar_scale_main}) need not be exponentially small in the total system size. If the cone grows to $\Theta(n)$ and the local two-copy contraction becomes Haar-like, the variance is exponentially suppressed. Thus, the relevant randomization is not global in the first instance; it is the randomization visible on the effective subsystem selected by the parameter and cost.


{\em Equal EP does not imply equal trainability.}---The hierarchy leads to a non-identifiability statement. Let $\mathcal{W}_{\bm{\theta}}(\hat{E})$ denote a fixed trainable wrapper into which a fixed entangling block $\hat{E}$ is inserted. The parameter locations, generators, input ensemble, and cost observable are held fixed; only the entangling block $\hat{E}$ is changed. For such a comparison, define the parameter-resolved gradient-variance profile
\begin{eqnarray}
\Gamma_{\hat{E}}(\bm{\theta}) = \left( \Var_\psi[g_1^{(\hat{E})}(\psi,\bm{\theta})],\ldots, \Var_\psi[g_p^{(\hat{E})}(\psi,\bm{\theta})] \right),
\label{eq:gamma_profile_main}
\end{eqnarray}
where $g_k^{(\hat{E})}$ is the gradient of the common cost with respect to the $k$th parameter after inserting $\hat{E}$.

\begin{theorem}[Non-identifiability from mean EP]
\label{thm:2}
There exists a finite-dimensional PQC architecture and a fixed cost observable for which the trainability profile does not factor through the scalar mean descriptor EP. Equivalently, there are two entangling blocks $\hat{E}$ and $\hat{E}'$ in the same wrapper such that
\begin{eqnarray}
\EP(\mathcal W_{\bm{\theta}}(\hat{E}))=\EP(\mathcal W_{\bm{\theta}}(\hat{E}')) \quad \text{for all}~\bm{\theta},
\label{eq:equal_ep_formal_main}
\end{eqnarray}
whereas, for the same parameter vector (indeed for all $\bm{\theta}$ in the constructive witness),
\begin{eqnarray}
\Gamma_{\hat{E}}(\bm{\theta}) \ne \Gamma_{\hat{E}'}(\bm{\theta}).
\label{eq:unequal_gamma_formal_main}
\end{eqnarray}
Consequently, no function $F:\mathbb R\to\mathbb R^p$ can satisfy $\Gamma_{\hat{E}}(\bm{\theta}) = F(\EP(\mathcal{W}_{\bm{\theta}}(\hat{E})))$ on this architecture class. The witness can also be chosen so that the two circuits have different EPDs.
\end{theorem}

\begin{proof}[Proof sketch]
Equality of EP fixes only the single global contraction $\Tr[\hat{K}_Q \hat{M}^{(2)}_{\hU}]$. In contrast, Eq.~(\ref{eq:grad_var_local_main}) shows that each gradient variance is a parameter- and cost-selected local contraction $\Tr[(\hat{B}_i \otimes \hat{B}_i) \hat{M}_{i,\LC}^{(2)}]$. These contractions need not be determined by the EP contraction. The statement is made constructive in Sec.~IV~B of the \SM, where two locally inequivalent entangling blocks with identical Meyer-Wallach EP but different EPD are inserted into the same one-parameter wrapper and yield different gradient variances for the same cost. The resulting non-factorization through EP is formulated abstractly in Sec.~IV~C of the \SM.
\end{proof}

Theorem~\ref{thm:2} rules out the strongest possible mean-descriptor picture. Equal EP means equal average entangling strength over product inputs; it does not fix the width of the output-entanglement distribution, nor does it fix how the local two-copy moment is oriented relative to a chosen generator and cost observable. Thus, two circuits can be indistinguishable by the coverage while lying in different trainability regimes.

This is also the role of EPD. EPD is not itself a barren-plateau invariant, because trainability remains local and cost-dependent. However, EPD is a strictly richer fixed-circuit descriptor than EP in the sense relevant here: it detects fluctuation structure that the mean discards and that can accompany a change in an actual gradient variance. 


\begin{figure*}[t]
\centering
\includegraphics[width=0.92\textwidth]{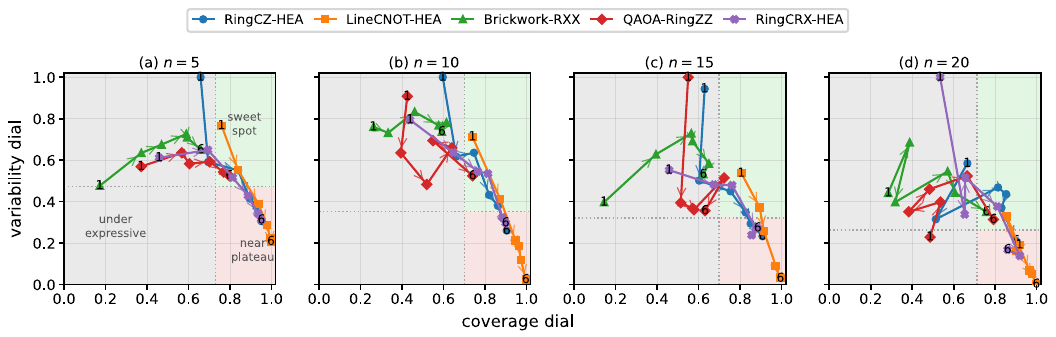}
\caption{Ansatz trajectories in the EP-EPD two-dial plane for $n=5,10,15,20$. Curves denote ansatz families; markers trace $L=1,\ldots,6$, with labels at the first and final depths and arrows in the increasing-$L$ direction. The axes are the coverage and variability dials in Eq.~(\ref{eq:two_dial_coords_main}). Gray, green, and red mark the underexpressive, sweet-spot, and near-plateau regimes in Eq.~(\ref{eq:sweet_region_main}), with thresholds fixed by Eq.~(\ref{eq:threshold_rule_main}).}
\label{fig:route}
\end{figure*}

{\em Two-dial routes for PQC design.}---We now use this hierarchy as a design protocol. For an ansatz family $A$, increasing depth $L$ traces a route whose two coordinates measure coverage and residual input-dependent variability. The objective is not to maximize depth, expressibility, or mean EP alone, but to locate a high-coverage, nonhomogenized window before the gradient-controlling local moments approach their Haar scale.

Fig.~\ref{fig:route} implements this EP/EPD scan for five PQC ansatz families: ring-CZ HEA, line-CNOT HEA, brickwork RXX, QAOA-style ring-ZZ, and ring-CRX HEA. For $n=5,10,15,20$ and $L=1,\ldots,6$, we estimate
\begin{eqnarray}
\overline{\EP}_{A,L}^{(n)} &=& \mathbb E_{\bm{\theta} \sim \nu_{A,L}^{(n)}} \bigl[ \EP(\hU_{A,L}^{(n)}(\bm{\theta})) \bigr], \nonumber \\
\overline{\EPD}_{A,L}^{(n)} &=&\mathbb E_{\bm{\theta} \sim \nu_{A,L}^{(n)}} \bigl[ \EPD(\hU_{A,L}^{(n)}(\bm{\theta})) \bigr].
\end{eqnarray}
Here, $\nu_{A,L}^{(n)}$ denotes the parameter-sampling measure on the depth-$L$ parameter space $\Theta_{A,L}^{(n)}$ for ansatz family $A$ and system size $n$; in the route scan it is the product distribution over the circuit angles used to sample fixed circuit instances, i.e., the family-specific version of the ensemble $\nu_L$ introduced above. The two plotted coordinates are
\begin{eqnarray}
x_{A,L}^{(n)}=\frac{\overline{\EP}_{A,L}^{(n)}}{Q_{\Haar}^{(n)}},
\qquad
y_{A,L}^{(n)}=\frac{\overline{\EPD}_{A,L}^{(n)}}{\max_{A,L}\overline{\EPD}_{A,L}^{(n)}},
\label{eq:two_dial_coords_main}
\end{eqnarray}
where $Q_{\Haar}^{(n)}=(2^n-2)/(2^n+1)$. Thus $x$ is a coverage dial and $y$ is the residual variability dial within the same $n$ panel. Projecting onto $x$ alone recovers the usual EP-vs-BP reading---larger mean entanglement as motion toward Haar-like behavior~\cite{McClean2018barrenplateaus,Cerezo2021BP,Holmes2022expressibility}---but cannot separate a useful high-coverage circuit from a homogenized one. The EPD dial resolves this ambiguity.

For reproducible region boundaries, we use the route data rather than hand-tuned constants. For a finite set $Z\subset[0,1]$, sort the distinct values and consider midpoint cuts $\tau$ between adjacent values. Each cut partitions $Z$ into $Z_{<\tau}$ and $Z_{\ge\tau}$, with means $\bar z_{<}$ and $\bar z_{\ge}$, and assigns
\begin{eqnarray}
W_Z(\tau) = \sum_{z \in Z_{<\tau}} \left( z - \bar{z}_{<} \right)^2 + \sum_{z \in Z_{\ge\tau}}\left( z - \bar{z}_{\ge} \right)^2.
\end{eqnarray}
We define $\tau_*(Z)=\arg\min_\tau W_Z(\tau)$, equivalently the Otsu/Fisher maximum-separation split~\cite{Otsu1979Threshold}, and set
\begin{eqnarray}
x_c^{(n)} = \tau_*\big(\{x_{A,L}^{(n)}\}_{A,L}\big),
\quad
y_c^{(n)} = \tau_*\big(Y_+^{(n)}\big),
\label{eq:threshold_rule_main}
\end{eqnarray}
where $Y_+^{(n)}=\{y_{A,L}^{(n)}:x_{A,L}^{(n)}\ge x_c^{(n)}\}$. The $y$ cut is conditional on coverage so that low variability is interpreted as plateau-like homogenization only after the circuit has passed the coverage cut. Details and thresholds are in Sec.~V-C of the \SM. We then define
\begin{eqnarray}
\mathcal S^{(n)}=\{(A,L):x_{A,L}^{(n)}\ge x_c^{(n)},
\;y_{A,L}^{(n)}\ge y_c^{(n)}\}.
\label{eq:sweet_region_main}
\end{eqnarray}
The green sweet spot in Fig.~\ref{fig:route} is $\mathcal S^{(n)}$: coverage has passed the cut while EPD remains open; gray is underexpressive, and red is high-coverage but homogenized. We rank candidates by
\begin{eqnarray}
s_{A,L}^{(n)}=x_{A,L}^{(n)}y_{A,L}^{(n)},
\label{eq:sweet_score_main}
\end{eqnarray}
which rewards simultaneous coverage and noncollapsed variability. Thresholds and best-score depths are in Table~I of Sec.~V~C of the \SM. Only the EPD dial separates the green high-coverage/nonhomogenized window from the red high-coverage/homogenized sector.

\begin{definition}[Two-dial trainability window]
For a chosen ansatz family and cost class, a practical trainability window is a depth interval in which $(A,L)\in\mathcal S^{(n)}$, or more generally in which the coverage is already high while $\overline{\EPD}_{A,L}^{(n)}$ has not yet collapsed relative to its early-depth scale.
\end{definition}

This definition is a screening principle, not a replacement for the local gradient criterion: the actual variance is still selected by the parameter light cone and cost observable in Eq.~(\ref{eq:grad_var_local_main}). The scan says where task-specific optimization is most worthwhile---away from the underexpressive gray and already homogenized red regions, and toward high-coverage points whose EPD dial remains open. 

A consistency check in Sec.~V-C of the \SM{} compares normalized EPD with representative gradient variance. The relation is not universal because Eq.~(\ref{eq:grad_var_local_main}) is local and cost-dependent, but low EPD broadly coincides with smaller gradients, especially at larger $n$.

\begin{figure}[t]
\centering
\includegraphics[width=0.46\textwidth]{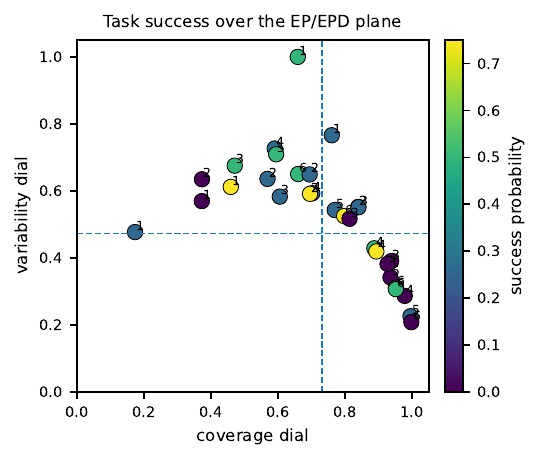}
\caption{Task-level validation on the $n=5$ route plane. Points are ansatz-depth pairs, color is empirical success probability, and dashed lines are $x_c,y_c$. High coverage with low EPD is less favorable for training.}
\label{fig:task_validation_scatter}
\end{figure}

\begin{figure}[t]
\centering
\includegraphics[width=0.46\textwidth]{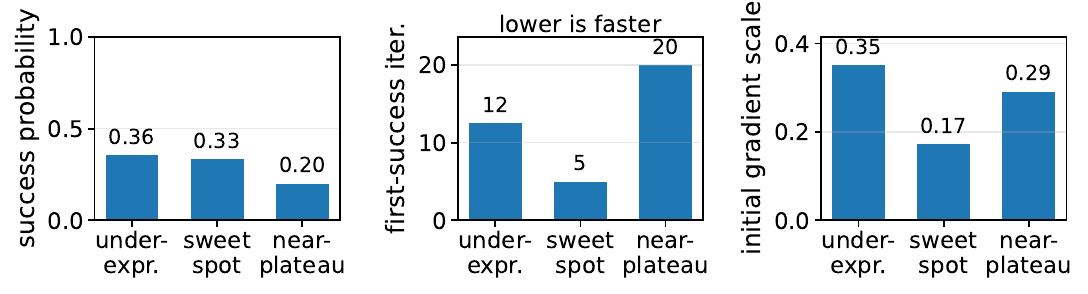}
\caption{Region-wise diagnostics: success probability, first-success iteration (lower is faster), and initial SPSA gradient scale. The sweet region is the fastest high-coverage region and outperforms near-plateau in success probability.}
\label{fig:task_validation_region}
\end{figure}

{\em Task-level validation.}---We finally test the descriptor-defined window in a descriptor-blind $n=5$ teacher--student QML benchmark. The EP/EPD coordinates, region labels, and sweet scores are fixed before training; all $30$ ansatz-depth candidates use the same embedding, readout, loss, and Adam-SPSA budget for two sweet-region teacher tasks. 

Figs.~\ref{fig:task_validation_scatter} and \ref{fig:task_validation_region} show that, within the high-coverage subset, the sweet region outperforms the low-EPD near-plateau sector in accuracy, loss, success probability, and first-success time, reaching the threshold after a median of $5$ iterations versus $20$ near the plateau. Thus, the sweet score can be a first-pass design rule. Full protocol, result tables, and analysis are in Sec.~VI of the \SM.

{\em Discussion.}---We have introduced an EP/EPD framework that separates three objects often compressed into a single expressibility--trainability narrative: EP is a \emph{global two-copy mean}, EPD is a \emph{global four-copy fluctuation} descriptor, and gradient variance is a \emph{local two-copy contraction} selected by a parameter light cone and cost observable. This hierarchy proves that equal EP does not imply equal EPD or equal gradient variance, and Fig.~\ref{fig:route} shows that Haar-like coverage and BP-like homogenization are distinct finite-depth crossovers.

The implication is that expressibility-vs-trainability is not a hard one-dimensional tradeoff. EP reports coverage, but only EPD tells whether high-coverage circuits still retain input-dependent variability. The EP/EPD scan therefore turns the usual EP-vs-BP tension into an ansatz-design rule: choose blocks, connectivities, and depth schedules that move far enough along the coverage dial while keeping the variability dial open. This matters for variational algorithms, hybrid heuristics, and quantum machine learning, where useful coverage must coexist with finite-shot gradients.

\begin{acknowledgments}
{Acknowledgments.}---This work was supported by the Ministry of Science and ICT/National Research Foundation of Korea (RS-2023-NR119924, RS-2024-00432214, RS-2025-03532992, and RS-2025-18362970), the IITP grant funded by the Korean government (RS-2019-II190003, RS-2026-25519864), and the Korean ARPA-H Project through KHIDI, funded by the Ministry of Health \& Welfare, Republic of Korea (RS-2025-25456722). We thank the Yonsei University Quantum Computing Project Group for support and access to Quantum System One.
\end{acknowledgments}

%


\clearpage
\newpage
\onecolumngrid

\makeatletter
\let\addcontentsline\origaddcontentsline
\makeatother

\part{Supplemental Material for ``Breaking the One-Dimensional Expressibility–Trainability Tradeoff''}

\tableofcontents

\section{Ensembles, fixed-circuit descriptors, and the spectral obstruction}\label{sec:setup}

This Supplemental Material follows the notation and proof of the original manuscript. Its purpose is threefold. First, it gives the detailed derivations behind the $2$-copy/$4$-copy hierarchy for fixed-circuit entangling power (EP)~\cite{SM-Zanardi2000EP} and entangling-power deviation (EPD)~\cite{SM-ChoBang2026EPD}. Second, it gives the local light-cone proof of the gradient-variance criterion used in the main manuscript. Third, it records the numerical protocols, supplementary diagnostics, and numerical analysis.

A useful way to read this Supplemental Material is to keep two orderings distinct. The first is the \emph{copy order} of the moment: EP is determined by a two-copy contraction, EPD by a four-copy contraction, and the gradient variance by a local two-copy contraction. The second is the \emph{ensemble order}: one may first average over product inputs for a fixed circuit and only then average over circuit parameters, or one may form Sim-style descriptors directly from parameter-sampled output states. These operations do not define the same object. The main manuscript uses this distinction to argue that a one-dimensional expressibility narrative is incomplete, while the present document supplies the algebra behind that statement.

Throughout, the $n$-qubit Hilbert space is
\begin{eqnarray}
\mathcal{H}_n=(\mathbb{C}^2)^{\otimes n},
\label{eq:sm_Hn}
\end{eqnarray}
and a depth-$L$ parameterized circuit is denoted by
\begin{eqnarray}
\hU_L(\bm{\theta})\in U(2^n),
\qquad \bm{\theta}\in\Theta_L\subset\mathbb R^{p_L}.
\label{eq:sm_pqc}
\end{eqnarray}
Here, $\bm{\theta}=(\theta_1,\ldots,\theta_{p_L})$ is a single parameter vector whose $p_L$ scalar components are the independent trainable parameters of the depth-$L$ ansatz. The number $p_L$ may depend on the depth, and $\Theta_L$ denotes the allowed joint domain of these parameters within the parameter space $\mathbb R^{p_L}$. Two randomness levels are used. The parameter ensemble $\nu_L$ samples a circuit instance and underlies Sim-style expressibility~\cite{SM-Sim2019expressibility}. The product-input ensemble used for fixed-circuit EP and EPD is instead
\begin{eqnarray}
\mu_{\mathrm{prod}}=\bigotimes_{j=1}^n\mu_{\mathrm H}^{(1)},
\qquad
\ket{\psi}=\bigotimes_{j=1}^n\ket{\psi_j},
\qquad
\ket{\psi_j}\sim\mu_{\mathrm H}^{(1)}.
\label{eq:sm_product_ensemble}
\end{eqnarray}
Thus, expressibility samples many circuit instances from a fixed reference input, whereas EP/EPD probe one fixed circuit with many product inputs.

For a pure state $\ket{\phi}$, the Meyer-Wallach global entanglement measure is~\cite{SM-MeyerWallach2002}
\begin{eqnarray}
Q(\ket{\phi})=\frac{2}{n}\sum_{j=1}^{n}\left(1-\Tr \hrho_j^2\right),
\qquad
\hrho_j=\Tr_{\bar j}\ketbra{\phi}{\phi}.
\label{eq:sm_Q_def}
\end{eqnarray}
Here, $\bar j:=\{1,\ldots,n\}\setminus\{j\}$ denotes the complement of qubit $j$. Thus, $\Tr_{\bar j}$ traces out all qubits except $j$ and leaves the one-qubit reduced density operator $\hrho_j$. Since $1/2\le \Tr \hrho_j^2\le1$, one has $0\le Q\le1$. Also, $Q(\ket{\phi})=0$ if and only if $\ket{\phi}$ is fully separable: if $Q=0$, every one-qubit reduction is pure; purity of a subsystem of a global pure state implies factorization across the corresponding bipartition; iterating over all sites gives full product structure.

For a fixed unitary $\hU$, we define
\begin{eqnarray}
\EP(\hU) := \mathbb E_{\psi\sim\mu_{\mathrm{prod}}}\left[Q(\hU\ket{\psi})\right],
\qquad
\EPD(\hU) := \sqrt{\Var_{\psi\sim\mu_{\mathrm{prod}}}\left[Q(\hU\ket{\psi})\right]}.
\label{eq:sm_EP_EPD_def}
\end{eqnarray}
If $\hat{V}_{\rm in}, \hat{V}_{\rm out}\in U(2)^{\otimes n}$ are local unitaries, then
\begin{eqnarray}
\EP(\hat{V}_{\rm out} \hU \hat{V}_{\rm in})=\EP(\hU),
\qquad
\EPD(\hat{V}_{\rm out} \hU \hat{V}_{\rm in})=\EPD(\hU).
\label{eq:sm_local_invariance}
\end{eqnarray}
\label{rem:LU_invariance}
The post-multiplication by $\hat{V}_{\rm out}$ preserves $Q$, while the pre-multiplication by $\hat{V}_{\rm in}$ preserves the product-Haar measure.

For $t\ge1$, we define the fixed-circuit output moment
\begin{eqnarray}
\hM^{(t)}_{\hU}:= \mathbb{E}_{\psi\sim\mu_{\mathrm{prod}}} \left[\left(\hU\ketbra{\psi}{\psi}\hU^\dagger\right)^{\otimes t}\right].
\label{eq:sm_MUt}
\end{eqnarray}
For one qubit,
\begin{eqnarray}
\int d\mu_{\mathrm H}^{(1)}(\psi)(\ketbra{\psi}{\psi})^{\otimes t} = \frac{\hat{\Pi}^{(t,1)}_{\rm sym}}{t+1},
\label{eq:sm_one_qubit_moment}
\end{eqnarray}
where $\hat{\Pi}^{(t,1)}_{\rm sym}$ projects onto $\mathrm{Sym}^t(\mathbb C^2)$. Grouping together, the $t$ copies of each physical qubit gives
\begin{eqnarray}
\hM^{(t)}_{\rm prod} = \frac{1}{(t+1)^n}\hat{\Pi}^{(t)}_{\rm loc},
\qquad
\hat{\Pi}^{(t)}_{\rm loc}=\bigotimes_{j=1}^n \hat{\Pi}^{(t)}_{{\rm sym},j}.
\label{eq:sm_prod_moment}
\end{eqnarray}
The full Haar pure-state moment on $d=2^n$ dimensions is
\begin{eqnarray}
\hM^{(t)}_{\rm Haar} =\frac{1}{D_{n,t}}\hat{\Pi}^{(t)}_{\rm glob},
\qquad
D_{n,t}=\binom{2^n+t-1}{t},
\label{eq:sm_haar_moment}
\end{eqnarray}
where $\hat{\Pi}^{(t)}_{\rm glob}$ projects onto $\mathrm{Sym}^t(\mathcal H_n)$.

\begin{proposition}[Spectral obstruction]\label{prop:spectral_obstruction}
For every fixed unitary $\hU$ and every $t\ge1$,
\begin{eqnarray}
\hM^{(t)}_{\hU}=\hU^{\otimes t}\hM^{(t)}_{\rm prod}\hU^{\dagger\otimes t}.
\label{eq:sm_spectral_orbit}
\end{eqnarray}
Therefore, the spectrum of $\hM^{(t)}_{\hU}$ is independent of $\hU$. If $n>1$ and $t\ge2$, then $\hM^{(t)}_{\rm prod}$ and $\hM^{(t)}_{\rm Haar}$ do not have the same spectrum.
\end{proposition}

\begin{proof}---Eq.~(\ref{eq:sm_spectral_orbit}) follows by linearity:
\begin{eqnarray}
\hM^{(t)}_{\hU} &=&\mathbb E_\psi\left[\hU^{\otimes t}(\ketbra{\psi}{\psi})^{\otimes t}\hU^{\dagger\otimes t}\right] =\hU^{\otimes t}\hM^{(t)}_{\rm prod}\hU^{\dagger\otimes t}.
\end{eqnarray}
Thus $\hM^{(t)}_{\hU}$ is unitarily conjugate to $\hM^{(t)}_{\rm prod}$.

By Eqs.~(\ref{eq:sm_prod_moment}) and (\ref{eq:sm_haar_moment}), the two reference moments are proportional to projectors with support dimensions
\begin{eqnarray}
\dim \left( \supp(\hM^{(t)}_{\rm prod}) \right) = (t+1)^n,
\quad
\dim \left( \supp(\hM^{(t)}_{\rm Haar}) \right) = \binom{2^n+t-1}{t}.
\label{eq:sm_rank_compare}
\end{eqnarray}
The locally symmetric subspace is strictly contained in the globally symmetric subspace when $n>1$ and $t\ge2$. To see strictness, take the vector obtained by symmetrizing one copy of $\ket{110\cdots0}$ among $t$ copies while all other copies are $\ket{0^n}$. This vector is globally copy-symmetric, but it is not invariant under swapping only the first-qubit factors of two copies. Hence, the supports, ranks, and spectra differ.
\end{proof}

This proposition is the reason that Haar is used in the main manuscript as a scalar entanglement reference and as a local-moment benchmark, not as an operator-level limit of fixed-circuit product-input moments. More generally, let $\mathcal H=\mathcal H_A\otimes\mathcal H_B$, with $d_A=\dim\mathcal H_A$ and $d_B=\dim\mathcal H_B$, and let $\hrho_A=\Tr_B\ketbra{\phi}{\phi}$ for a Haar-random $\ket{\phi}\in\mathcal H$. Lubkin's formula states $\mathbb E_{\phi\sim\Haar}[\Tr(\hrho_A^2)]=(d_A+d_B)/(d_Ad_B+1)$~\cite{SM-Lubkin1978}. Indeed, the Haar two-copy moment $\mathbb E[\ketbra{\phi}{\phi}^{\otimes2}]=(\hid+\hS_A\hS_B)/[d_Ad_B(d_Ad_B+1)]$, together with $\Tr(\hrho_A^2)=\Tr[\hS_A\ketbra{\phi}{\phi}^{\otimes2}]$, gives this result by a single swap contraction. Taking $d_A=2$ and $d_B=2^{n-1}$ yields the one-qubit-versus-rest value
\begin{eqnarray}
\mathbb E_{\phi\sim\Haar}\left[\Tr\hrho_j^2\right] =\frac{2+2^{n-1}}{2^n+1}.
\label{eq:sm_lubkin}
\end{eqnarray}
Substituting this into Eq.~(\ref{eq:sm_Q_def}) gives
\begin{eqnarray}
Q_{\Haar} = \mathbb E_{\phi\sim\Haar}[Q(\ket{\phi})] = \frac{2^n-2}{2^n+1}.
\label{eq:sm_Q_Haar}
\end{eqnarray}

\section{Exact two-copy and four-copy formulas for EP and EPD}\label{sec:moment}

Let $\hS_j$ denote the unitary that swaps qubit $j$ between two copies and acts trivially elsewhere. The following identity is repeatedly used to convert local purities into two-copy expectation values. If $\hA,\hB$ are operators on $\mathcal H_n$, then
\begin{eqnarray}
\Tr [\hS_j(\hA\otimes \hB)] = \Tr\left[(\Tr_{\bar j}\hA)(\Tr_{\bar j}\hB)\right].
\label{eq:sm_swap_identity}
\end{eqnarray}
To verify Eq.~(\ref{eq:sm_swap_identity}) explicitly, write basis states as $\ket{a,\mu}=\ket{a}_j\ket{\mu}_{\bar j}$ (or $\ket{b,\nu}=\ket{b}_j\ket{\nu}_{\bar j}$), where $a,b\in\{0,1\}$ are single-qubit labels and $\mu,\nu$ label an orthonormal basis of the complement. The swap acts as
\begin{eqnarray}
\hS_j(\ket{a,\mu}\otimes\ket{b,\nu})=\ket{b,\mu}\otimes\ket{a,\nu}.
\end{eqnarray}
Therefore,
\begin{eqnarray}
\Tr[\hS_j(\hA\otimes \hB)] &=&\sum_{a,b,\mu,\nu} \left(\bra{a,\mu}\otimes\bra{b,\nu}\right) \hS_j (\hA\otimes \hB) \left(\ket{a,\mu}\otimes\ket{b,\nu}\right)
\nonumber\\
&=&\sum_{a,b,\mu,\nu}
\bra{b,\mu}\hA\ket{a,\mu}
\bra{a,\nu}\hB\ket{b,\nu}
\nonumber\\
&=&\sum_{a,b}\bra{b}(\Tr_{\bar j}\hA)\ket{a}\bra{a}(\Tr_{\bar j}\hB)\ket{b}
\nonumber\\
&=&\Tr[(\Tr_{\bar j}\hA)(\Tr_{\bar j}\hB)].
\end{eqnarray}
For $\hrho=\ketbra{\phi}{\phi}$, Eq.~(\ref{eq:sm_swap_identity}) yields
\begin{eqnarray}
\Tr[\hS_j\hrho^{\otimes2}]=\Tr\hrho_j^2.
\end{eqnarray}
Therefore, with
\begin{eqnarray}
\hK_Q=\frac{2}{n}\sum_{j=1}^n(\hid-\hS_j),
\label{eq:sm_KQ}
\end{eqnarray}
one has the exact two-copy expression
\begin{eqnarray}
Q(\ket{\phi})=\Tr[\hK_Q\hrho^{\otimes2}].
\label{eq:sm_Q_two_copy}
\end{eqnarray}
On four copies, let $\hK_Q^{(12)}$ and $\hK_Q^{(34)}$ act on copy pairs $(1,2)$ and $(3,4)$, respectively. Define
\begin{eqnarray}
\hT_Q=\hK_Q^{(12)}\hK_Q^{(34)}
=\frac{4}{n^2}\sum_{j,k=1}^n(\hid-\hS_j^{(12)})(\hid-\hS_k^{(34)}).
\label{eq:sm_TQ}
\end{eqnarray}
Since $\hrho^{\otimes4}=\hrho^{\otimes2}_{12}\otimes\hrho^{\otimes2}_{34}$,
\begin{eqnarray}
Q(\ket{\phi})^2=\Tr[\hT_Q\hrho^{\otimes4}].
\label{eq:sm_Q_four_copy}
\end{eqnarray}

\begin{theorem}[Exact moment representation]
For a fixed circuit $\hU$,
\begin{eqnarray}
\EP(\hU)=\Tr[\hK_Q\hM^{(2)}_{\hU}],
\label{eq:sm_EP_two_copy}
\end{eqnarray}
and
\begin{eqnarray}
\EPD(\hU)^2
=\Tr[\hT_Q\hM^{(4)}_{\hU}]-\left(\Tr[\hK_Q\hM^{(2)}_{\hU}]\right)^2.
\label{eq:sm_EPD_four_copy}
\end{eqnarray}
\end{theorem}

\begin{proof}---For $\hrho_\psi(\hU)=\hU\ketbra{\psi}{\psi}\hU^\dagger$, Eq.~(\ref{eq:sm_Q_two_copy}) gives
\begin{eqnarray}
Q(\hU\ket{\psi})=\Tr[\hK_Q\hrho_\psi(\hU)^{\otimes2}].
\end{eqnarray}
Averaging and interchanging trace and expectation gives Eq.~(\ref{eq:sm_EP_two_copy}). Similarly, Eq.~(\ref{eq:sm_Q_four_copy}) gives
\begin{eqnarray}
\mathbb E_\psi[Q(\hU\ket{\psi})^2]=\Tr[\hT_Q\hM^{(4)}_{\hU}].
\end{eqnarray}
Subtracting the square of the mean gives Eq.~(\ref{eq:sm_EPD_four_copy}).
\end{proof}

The theorem is the precise sense in which EP and EPD probe different moment orders. Equality of EP values constrains only one scalar contraction of a two-copy moment and does not determine the four-copy fluctuation profile.

For completeness, we record the stability bounds used implicitly in the main manuscript. For two circuits $\hat U,\hat V$,
\begin{eqnarray}
\abs{\EP(\hat U)-\EP(\hat V)}
\le \norm{\hK_Q}_2\norm{\hM^{(2)}_{\hat U}-\hM^{(2)}_{\hat V}}_2,
\label{eq:sm_EP_stability}
\end{eqnarray}
and
\begin{eqnarray}
\abs{\EPD(\hat U)^2-\EPD(\hat V)^2}
\le
\norm{\hT_Q}_2\norm{\hM^{(4)}_{\hat U}-\hM^{(4)}_{\hat V}}_2
+2\norm{\hK_Q}_2\norm{\hM^{(2)}_{\hat U}-\hM^{(2)}_{\hat V}}_2.
\label{eq:sm_EPD_stability}
\end{eqnarray}
The first inequality is Hilbert-Schmidt Cauchy-Schwarz. The second follows by applying the same inequality to the four-copy term and bounding the change in the squared mean by $2\abs{\EP(\hat U)-\EP(\hat V)}$, since $0\le\EP\le1$.

Expanding $\Tr \hK_Q^2$ gives
\begin{eqnarray}
\norm{\hK_Q}_2^2=4^n\frac{n+3}{n},
\qquad
\norm{\hT_Q}_2=4^n\frac{n+3}{n}.
\label{eq:sm_norms}
\end{eqnarray}
Indeed,
\begin{eqnarray}
\norm{\hK_Q}_2^2
=\Tr \hK_Q^2
=\frac{4}{n^2}\sum_{j,k=1}^n\Tr[(\hid-\hS_j)(\hid-\hS_k)].
\end{eqnarray}
If $j=k$, then $(\hid-\hS_j)^2=2(\hid-\hS_j)$ and
\begin{eqnarray}
\Tr[(\hid-\hS_j)^2]=2\Tr(\hid-\hS_j)=2(4^n-2\cdot4^{n-1})=4^n.
\end{eqnarray}
If $j\ne k$, the trace factorizes over the two nontrivial site pairs and the remaining $n-2$ sites:
\begin{eqnarray}
\Tr[(\hid-\hS_j)(\hid-\hS_k)] &=& \Tr[\hid-\hS_j-\hS_k+\hS_j\hS_k] \nonumber\\
	&=& 4^n-2\cdot4^{n-1}-2\cdot4^{n-1}+2\cdot2\cdot4^{n-2} \nonumber\\
	&=& 2\cdot2\cdot4^{n-2}=4^{n-1}.
\end{eqnarray}
Therefore,
\begin{eqnarray}
\norm{\hK_Q}_2^2=\frac{4}{n^2}\left(n4^n+n(n-1)4^{n-1}\right)=4^n\frac{n+3}{n}.
\end{eqnarray}
Because $\hT_Q=\hK_Q\otimes \hK_Q$ under the copy grouping $(12)|(34)$, the Hilbert-Schmidt norm is multiplicative and $\norm{\hT_Q}_2=\norm{\hK_Q}_2^2$.

\section{Gradient Variance and the Local 2-Copy Criterion for BP}\label{sec:gradientBP}

We now turn to trainability. The main purpose of this section is to isolate, in exact algebraic form, the object that controls the second moment of gradients. This is where the usual expressibility--vs--barren plateau (BP) narrative must be sharpened. The quantity $\mathrm{EP}(\hU)$ derived in Sec.~\ref{sec:moment} is a \emph{global} two-copy scalar functional of a fixed circuit, while $\mathrm{EPD}(\hU)$ is a \emph{global} four-copy scalar functional. By contrast, the variance of the gradient with respect to a specific parameter is controlled not by any global scalar descriptor, but by a \emph{local} two-copy moment associated with the parameter's effective light cone.

This is the precise sense in which the barren plateaus are fundamentally a \emph{local second-moment phenomenon}. In the language of McClean-type results, one often says that the gradients vanish when the relevant circuit becomes sufficiently random or design-like~\cite{SM-McClean2018barrenplateaus,SM-Cerezo2022challenges}. What the present section makes explicit is that the operative notion is not a global one-number summary such as expressibility or average entangling power, but the second-copy structure seen by each parameter through its causal cone.

Throughout this section, we consider a differentiable $n$-qubit PQC
\begin{eqnarray}
\hU(\bm{\theta}) = \hU_m(\bm{\theta})\cdots \hU_2(\bm{\theta})\hU_1(\bm{\theta}),
\label{eq:BP_circuit_seq}
\end{eqnarray}
built from geometrically local gates on a fixed connectivity graph. For notational clarity, we assume that each parameter $\theta_i$ appears in a single elementary gate. If the same parameter controls several gates, the total gradient is the sum of the corresponding single-gate contributions, and all arguments below apply termwise.

We study a cost function of the form
\begin{eqnarray}
C(\bm{\theta},\psi)
=
\Tr\left[
\hO\,\hrho_{\psi}(\bm{\theta})
\right],
\qquad
\hrho_{\psi}(\bm{\theta})
:=
\hU(\bm{\theta})\ketbra{\psi}{\psi}\hU(\bm{\theta})^{\dagger},
\label{eq:BP_cost}
\end{eqnarray}
where $\hO=\hO^{\dagger}$ is a bounded observable, which we normalize as
\begin{eqnarray}
\norm{\hO}_{\infty}\le 1.
\end{eqnarray}
The input state is again sampled from the Haar-product ensemble introduced in Sec.~\ref{sec:setup},
\begin{eqnarray}
\ket{\psi}
=
\bigotimes_{j=1}^{n}\ket{\psi_j},
\qquad
\ket{\psi_j}\sim \mu_{\Haar}^{(1)}
\quad
\text{independently}.
\label{eq:BP_input}
\end{eqnarray}
For a fixed parameter index $i$, we denote the corresponding gradient by
\begin{eqnarray}
g_i(\psi)
:=
\partial_{\theta_i} C(\bm{\theta},\psi).
\label{eq:def_gi}
\end{eqnarray}

\subsection{Dressed generators and effective light cones}\label{subsec:BP_lightcone}

Let the $i$th parameter enter the circuit through a gate of the form
\begin{eqnarray}
\hU_i(\theta_i)=e^{-i\theta_i \hG_i},
\label{eq:param_gate}
\end{eqnarray}
where $\hG_i=\hG_i^{\dagger}$ is a Hermitian generator supported on a bounded set of qubits $\Lambda_i$. We decompose the total circuit around this gate as
\begin{eqnarray}
\hU(\bm{\theta})
=
\hU_i^{\rm post}(\bm{\theta})\,e^{-i\theta_i \hG_i}\,\hU_i^{\rm pre}(\bm{\theta}),
\label{eq:circuit_split}
\end{eqnarray}
where $\hU_i^{\rm pre}$ collects all gates before the parameterized gate and $\hU_i^{\rm post}$ collects all gates after it, matching the pre/post notation used in the main manuscript.

The key output-picture object is the dressed generator
\begin{eqnarray}
\hH_i^{\mathrm{out}}
:=
i\,
\left(
\partial_{\theta_i}\hU(\bm{\theta})
\right)
\hU(\bm{\theta})^{\dagger}.
\label{eq:def_Hout}
\end{eqnarray}
Using Eq.~(\ref{eq:circuit_split}), we obtain the exact identity
\begin{eqnarray}
\hH_i^{\mathrm{out}}
=
\hU_i^{\rm post}(\bm{\theta})\,\hG_i\,\hU_i^{\rm post}(\bm{\theta})^{\dagger}.
\label{eq:Hout_explicit}
\end{eqnarray}
In particular, $\hH_i^{\mathrm{out}}$ is Hermitian and
\begin{eqnarray}
\norm{\hH_i^{\mathrm{out}}}_{\infty}
=
\norm{\hG_i}_{\infty}.
\label{eq:Hout_norm}
\end{eqnarray}

Eq.~(\ref{eq:Hout_explicit}) shows that only the gates \emph{after} the parameter dress the generator in the output picture. Because the circuit is local, the support of $\hH_i^{\mathrm{out}}$ can be obtained recursively by propagating the support $\Lambda_i$ through later gates: conjugation by a gate disjoint from the current support leaves the support unchanged, whereas conjugation by a gate overlapping the current support enlarges it at most to the union of the two supports. The resulting output subset is the usual forward light cone of the parameter. We denote it by
\begin{eqnarray}
\mathcal{F}(i)
:=
\operatorname{supp}\left(\hH_i^{\mathrm{out}}\right).
\label{eq:def_forward_cone}
\end{eqnarray}

Taking the derivative of the output state gives
\begin{eqnarray}
\partial_{\theta_i}\hrho_{\psi}(\bm{\theta})
&=&
\left(\partial_{\theta_i}\hU\right)\ketbra{\psi}{\psi}\hU^{\dagger}
+
\hU\ketbra{\psi}{\psi}\left(\partial_{\theta_i}\hU^{\dagger}\right)
\nonumber\\
&=&
-i\hH_i^{\mathrm{out}}\hrho_{\psi}(\bm{\theta})
+
i\hrho_{\psi}(\bm{\theta})\hH_i^{\mathrm{out}}
\nonumber\\
&=&
-i\left[
\hH_i^{\mathrm{out}},
\hrho_{\psi}(\bm{\theta})
\right].
\label{eq:rho_derivative_comm}
\end{eqnarray}
Hence, the gradient can be written as
\begin{eqnarray}
g_i(\psi)
&=&
\Tr\left[
\hO\,\partial_{\theta_i}\hrho_{\psi}(\bm{\theta})
\right]
=
-i\Tr\left[
\hO\left(\hH_i^{\mathrm{out}}\hrho_{\psi}-\hrho_{\psi}\hH_i^{\mathrm{out}}\right)
\right]
\nonumber\\
&=&
\Tr\left[
i\left[\hH_i^{\mathrm{out}},\hO\right]\hrho_{\psi}(\bm{\theta})
\right].
\label{eq:gradient_commutator_output}
\end{eqnarray}

This motivates the following definition.

\begin{definition}[Effective gradient observable and effective light cone]\label{def:effective_lightcone}
For the parameter $\theta_i$, we define the effective gradient observable
\begin{eqnarray}
\hB_i
:=
i\left[
\hH_i^{\mathrm{out}},\hO
\right].
\label{eq:def_Bi}
\end{eqnarray}
Its support
\begin{eqnarray}
\mathrm{LC}(i)
:=
\operatorname{supp}(\hB_i)
\label{eq:def_LCi}
\end{eqnarray}
is called the effective light cone of the parameter with respect to the chosen cost operator. We denote its size by
\begin{eqnarray}
r_i := |\mathrm{LC}(i)|,
\qquad
d_i := 2^{r_i}.
\label{eq:def_ri_di}
\end{eqnarray}
\end{definition}

By construction,
\begin{eqnarray}
\mathrm{LC}(i)
\subseteq
\mathcal{F}(i)\cup \operatorname{supp}(\hO).
\label{eq:LC_subset}
\end{eqnarray}
In particular, if $\hO$ is local and its support is disjoint from the forward light cone $\mathcal{F}(i)$, then $[\hH_i^{\mathrm{out}},\hO]=0$ and the corresponding gradient vanishes identically. Thus, the relevant subsystem for trainability is not the entire circuit, but the causal overlap between the parameter and the cost operator.

The observable $\hB_i$ inherits several basic structural properties from Eq.~(\ref{eq:def_Bi}).

\begin{lemma}[Basic properties of $\hB_i$]\label{lem:Bi_basic}
For each parameter index $i$, the operator $\hB_i$ defined in Eq.~(\ref{eq:def_Bi}) is Hermitian, traceless, and satisfies
\begin{eqnarray}
\norm{\hB_i}_{\infty}
\le
2\norm{\hG_i}_{\infty}\norm{\hO}_{\infty}.
\label{eq:Bi_norm_bound}
\end{eqnarray}
\end{lemma}

\begin{proof}---Hermiticity follows from
\begin{eqnarray}
\hB_i^{\dagger}
=
\left(
i[\hH_i^{\mathrm{out}},\hO]
\right)^{\dagger}
=
-i[\hO,\hH_i^{\mathrm{out}}]
=
i[\hH_i^{\mathrm{out}},\hO]
=
\hB_i,
\end{eqnarray}
since both $\hH_i^{\mathrm{out}}$ and $\hO$ are Hermitian. Tracelessness follows from the cyclicity of trace:
\begin{eqnarray}
\Tr(\hB_i) = i\Tr\left(\hH_i^{\mathrm{out}}\hO-\hO\hH_i^{\mathrm{out}}\right) = 0.
\end{eqnarray}
Finally, we have
\begin{eqnarray}
\norm{\hB_i}_{\infty}
=
\norm{i[\hH_i^{\mathrm{out}},\hO]}_{\infty}
\le
\norm{\hH_i^{\mathrm{out}}\hO}_{\infty}
+
\norm{\hO\hH_i^{\mathrm{out}}}_{\infty}
\le
2\norm{\hH_i^{\mathrm{out}}}_{\infty}\norm{\hO}_{\infty}
=
2\norm{\hG_i}_{\infty}\norm{\hO}_{\infty},
\end{eqnarray}
where Eq.~(\ref{eq:Hout_norm}) was used in the last step.
\end{proof}

\subsection{Exact local 2-copy formula for gradient moments}\label{subsec:BP_local_moment}

We now prove the central structural statement of this section. The gradient associated with a given parameter is a linear functional of the reduced output state on its effective light cone, and its second moment is determined exactly by the corresponding local two-copy moment.

\begin{proposition}[Light-cone reduction of gradient moments]\label{prop:lightcone_reduction}
Let $g_i(\psi)=\partial_{\theta_i}C(\bm{\theta},\psi)$ be the gradient with respect to the parameter $\theta_i$, and let $\hB_i$ and $\mathrm{LC}(i)$ be as in Definition~\ref{def:effective_lightcone}. Define the reduced output state on the effective light cone by
\begin{eqnarray}
\hsigma_{\psi}^{(i)}
:=
\Tr_{\overline{\mathrm{LC}(i)}}
\left[
\hrho_{\psi}(\bm{\theta})
\right].
\label{eq:def_sigma_i}
\end{eqnarray}
Then,
\begin{eqnarray}
g_i(\psi)
=
\Tr\left[
\hB_i\,\hsigma_{\psi}^{(i)}
\right].
\label{eq:gradient_local_linear}
\end{eqnarray}

Moreover, if we define the local two-copy moment
\begin{eqnarray}
\hM_{i,\mathrm{LC}}^{(2)}
:=
\mathbb{E}_{\psi\sim\mu_{\mathrm{prod}}}
\left[
\left(
\hsigma_{\psi}^{(i)}
\right)^{\otimes 2}
\right],
\label{eq:def_local_moment2}
\end{eqnarray}
then
\begin{eqnarray}
\mathbb{E}_{\psi}\left[g_i(\psi)^2\right]
=
\Tr\left[
(\hB_i\otimes \hB_i)\,
\hM_{i,\mathrm{LC}}^{(2)}
\right].
\label{eq:gradient_second_moment_local}
\end{eqnarray}
In fact, $\mathbb{E}_{\psi}[g_i(\psi)]=0$, and therefore
\begin{eqnarray}
\Var_{\psi}[g_i(\psi)]
=
\Tr\left[
(\hB_i\otimes \hB_i)\,
\hM_{i,\mathrm{LC}}^{(2)}
\right].
\label{eq:gradient_variance_local}
\end{eqnarray}
\end{proposition}

\begin{proof}---Eq.~(\ref{eq:gradient_local_linear}) follows directly from Eq.~(\ref{eq:gradient_commutator_output}) and the fact that $\hB_i$ is supported on $\mathrm{LC}(i)$. Indeed,
\begin{eqnarray}
g_i(\psi)
=
\Tr\left[
\hB_i\,\hrho_{\psi}(\bm{\theta})
\right]
=
\Tr\left[
\hB_i\,
\Tr_{\overline{\mathrm{LC}(i)}}
\left(
\hrho_{\psi}(\bm{\theta})
\right)
\right]
=
\Tr\left[
\hB_i\,\hsigma_{\psi}^{(i)}
\right].
\end{eqnarray}

To compute the second moment, square Eq.~(\ref{eq:gradient_local_linear}) and use the standard identity
\begin{eqnarray}
\Tr[\hat{X}\hat{\omega}]\Tr[\hat{Y}\hat{\omega}]
=
\Tr\left[
(\hat{X} \otimes \hat{Y})\,
\hat{\omega}^{\otimes 2}
\right]
\label{eq:trace_square_identity}
\end{eqnarray}
for any operators $\hat{X}, \hat{Y}, \hat{\omega}$ on the same Hilbert space. With $\hat{X} = \hat{Y} = \hB_i$ and $\hat{\omega}=\hsigma_{\psi}^{(i)}$, this gives
\begin{eqnarray}
g_i(\psi)^2
=
\Tr\left[
(\hB_i\otimes \hB_i)\,
\left(\hsigma_{\psi}^{(i)}\right)^{\otimes 2}
\right].
\end{eqnarray}
Averaging over $\psi$ and using linearity of trace and expectation,
\begin{eqnarray}
\mathbb{E}_{\psi}[g_i(\psi)^2]
&=&
\Tr\left[
(\hB_i\otimes \hB_i)\,
\mathbb{E}_{\psi}
\left[
\left(\hsigma_{\psi}^{(i)}\right)^{\otimes 2}
\right]
\right]
\nonumber\\
&=&
\Tr\left[
(\hB_i\otimes \hB_i)\,
\hM_{i,\mathrm{LC}}^{(2)}
\right],
\end{eqnarray}
which proves Eq.~(\ref{eq:gradient_second_moment_local}).

It remains to show that the mean gradient vanishes under the Haar-product input ensemble. First note that
\begin{eqnarray}
\mathbb{E}_{\psi\sim\mu_{\mathrm{prod}}}
\left[
\ketbra{\psi}{\psi}
\right]
=
\left(
\frac{\hid}{2}
\right)^{\otimes n}
=
\frac{\hid}{2^n}.
\label{eq:avg_input_maxmixed}
\end{eqnarray}
Therefore,
\begin{eqnarray}
\mathbb{E}_{\psi}
\left[
\hrho_{\psi}(\bm{\theta})
\right]
=
\hU(\bm{\theta})
\frac{\hid}{2^n}
\hU(\bm{\theta})^{\dagger}
=
\frac{\hid}{2^n}.
\label{eq:avg_output_maxmixed}
\end{eqnarray}
Taking the partial trace over the complement of $\mathrm{LC}(i)$,
\begin{eqnarray}
\mathbb{E}_{\psi}
\left[
\hsigma_{\psi}^{(i)}
\right]
=
\frac{\hid_{\mathrm{LC}(i)}}{2^{r_i}}
=
\frac{\hid_{\mathrm{LC}(i)}}{d_i}.
\label{eq:avg_sigma_maxmixed}
\end{eqnarray}
Hence,
\begin{eqnarray}
\mathbb{E}_{\psi}[g_i(\psi)]
&=&
\Tr\left[
\hB_i\,
\mathbb{E}_{\psi}[\hsigma_{\psi}^{(i)}]
\right]
\nonumber\\
&=&
\frac{1}{d_i}\Tr(\hB_i)
=
0,
\end{eqnarray}
where Lemma~\ref{lem:Bi_basic} was used in the last step. Consequently,
\begin{eqnarray}
\Var_{\psi}[g_i(\psi)]
=
\mathbb{E}_{\psi}[g_i(\psi)^2],
\end{eqnarray}
and Eq.~(\ref{eq:gradient_variance_local}) follows from Eq.~(\ref{eq:gradient_second_moment_local}).
\end{proof}

Proposition~\ref{prop:lightcone_reduction} is the exact local statement underlying this paper's trainability narrative. The relevant object is neither the global expressibility of the full ansatz nor the global EP of the full circuit. For each parameter, the gradient variance is determined by a specific local two-copy moment, namely $\hM_{i,\mathrm{LC}}^{(2)}$, contracted against the parameter-dependent observable $\hB_i\otimes \hB_i$.

In this sense, the variance of the gradient is a \emph{local} analogue of the global two-copy formula for EP. The crucial difference is that the object depends on the parameter index $i$ and on the effective light cone induced by the cost operator. This is exactly why one should not expect a single global scalar descriptor to reconstruct the full trainability profile $\{\Var(g_i)\}_i$.

\subsection{Local Haar benchmark and barren-plateau scaling}\label{subsec:BP_Haar_benchmark}

We next compare the exact local second moment to a Haar benchmark on the $r_i$-qubit effective light-cone Hilbert space. As emphasized already in Sec.~\ref{sec:setup}, this benchmark should not be interpreted as an operator-level convergence claim for the fixed-circuit/product-input ensemble. Rather, it is a scalar reference value for the contraction relevant to gradient variance.

Let $\mu_{\Haar}^{(r_i)}$ denote the Haar measure on pure states of the $r_i$-qubit Hilbert space associated with $\mathrm{LC}(i)$, and let
\begin{eqnarray}
\hM_{\mathrm{Haar},r_i}^{(2)}
:=
\mathbb{E}_{\phi\sim\mu_{\Haar}^{(r_i)}}
\left[
\ketbra{\phi}{\phi}^{\otimes 2}
\right]
=
\frac{\hid+\hS_i}{d_i(d_i+1)},
\label{eq:local_Haar_moment}
\end{eqnarray}
where $\hS_i$ is the full swap operator on two copies of the $r_i$-qubit light-cone Hilbert space.

We first compute the Haar benchmark exactly for any traceless local observable.

\begin{lemma}[Haar benchmark for linear observables]\label{lem:Haar_gradient_benchmark}
Let $\hB=\hB^{\dagger}$ be an operator on a $d$-dimensional Hilbert space with $\Tr(\hB)=0$. For a Haar-random pure state $\ket{\phi}$, define
\begin{eqnarray}
h_{\hB}(\phi)
:=
\Tr\left[
\hB\,\ketbra{\phi}{\phi}
\right].
\label{eq:def_hB}
\end{eqnarray}
Then, we have
\begin{eqnarray}
\mathbb{E}_{\phi}[h_{\hB}(\phi)] = 0, 
\qquad
\Var_{\phi}[h_{\hB}(\phi)] = \frac{\Tr(\hB^2)}{d(d+1)}.
\label{eq:hB_var_exact}
\end{eqnarray}
In particular,
\begin{eqnarray}
\Var_{\phi}[h_{\hB}(\phi)]
\le
\frac{\norm{\hB}_{\infty}^{2}}{d+1}.
\label{eq:hB_var_bound}
\end{eqnarray}
\end{lemma}

\begin{proof}---The mean is immediate from the Haar first moment:
\begin{eqnarray}
\mathbb{E}_{\phi}\left[\ketbra{\phi}{\phi}\right]
=
\frac{\hid}{d},
\end{eqnarray}
hence
\begin{eqnarray}
\mathbb{E}_{\phi}[h_{\hB}(\phi)]
=
\frac{\Tr(\hB)}{d}
=
0.
\end{eqnarray}

For the second moment, use Eq.~(\ref{eq:trace_square_identity}) with $\hat{\omega}=\ketbra{\phi}{\phi}$ and $\hat{X}=\hat{Y}=\hB$:
\begin{eqnarray}
h_{\hB}(\phi)^2
=
\Tr\left[
(\hB\otimes \hB)
\ketbra{\phi}{\phi}^{\otimes 2}
\right].
\end{eqnarray}
Averaging over $\phi$ gives
\begin{eqnarray}
\mathbb{E}_{\phi}[h_{\hB}(\phi)^2]
&=&
\Tr\left[
(\hB\otimes \hB)\,
\hM_{\mathrm{Haar}}^{(2)}
\right]
=
\frac{1}{d(d+1)}
\Tr\left[
(\hB\otimes \hB)(\hid+\hS)
\right]
\nonumber\\
&=&
\frac{1}{d(d+1)}
\left(
\Tr(\hB)^2+\Tr(\hB^2)
\right)
=
\frac{\Tr(\hB^2)}{d(d+1)}.
\end{eqnarray}
Since the mean is zero, this is the variance, proving Eq.~(\ref{eq:hB_var_exact}). Finally,
\begin{eqnarray}
\Tr(\hB^2)\le d\norm{\hB}_{\infty}^{2},
\end{eqnarray}
which immediately yields Eq.~(\ref{eq:hB_var_bound}).
\end{proof}

Applying the lemma to the operator $\hB_i$ from Proposition~\ref{prop:lightcone_reduction} yields the local Haar benchmark for gradients.

\begin{corollary}[Local 2-copy criterion for barren-plateau scaling]\label{cor:local2copy_BP}
Let $g_i(\psi)$ be the gradient associated with the parameter $\theta_i$, and let $\hB_i$, $r_i$, $d_i$, and $\hM_{i,\mathrm{LC}}^{(2)}$ be as in Proposition~\ref{prop:lightcone_reduction}. Define the Haar benchmark variance on the $r_i$-qubit light-cone Hilbert space by
\begin{eqnarray}
\Var_{\mathrm{Haar},r_i}[g_i] := \Var_{\phi\sim\mu_{\Haar}^{(r_i)}} \left[ \Tr\left(\hB_i\ketbra{\phi}{\phi}\right) \right] = \Tr\left[(\hB_i\otimes\hB_i)\,\hM_{\mathrm{Haar},r_i}^{(2)}\right].
\label{eq:def_local_Haar_var}
\end{eqnarray}
Then,
\begin{eqnarray}
\Var_{\psi}[g_i(\psi)]
=
\Var_{\mathrm{Haar},r_i}[g_i]
+
\Tr\left[
(\hB_i\otimes \hB_i)
\left(
\hM_{i,\mathrm{LC}}^{(2)}-\hM_{\mathrm{Haar},r_i}^{(2)}
\right)
\right].
\label{eq:var_difference_exact}
\end{eqnarray}
Consequently,
\begin{eqnarray}
\left|
\Var_{\psi}[g_i(\psi)]
-
\Var_{\mathrm{Haar},r_i}[g_i]
\right|
\le
\norm{\hB_i}_{2}^{2}\,
\norm{
\hM_{i,\mathrm{LC}}^{(2)}-\hM_{\mathrm{Haar},r_i}^{(2)}
}_{2}.
\label{eq:var_difference_bound}
\end{eqnarray}
Moreover,
\begin{eqnarray}
\Var_{\mathrm{Haar},r_i}[g_i]
=
\frac{\Tr(\hB_i^2)}{d_i(d_i+1)}
\le
\frac{\norm{\hB_i}_{\infty}^{2}}{d_i+1}
\le
\frac{4\norm{\hG_i}_{\infty}^{2}\norm{\hO}_{\infty}^{2}}{2^{r_i}+1}.
\label{eq:local_Haar_BP_scale}
\end{eqnarray}
Therefore, whenever the correction term in Eq.~(\ref{eq:var_difference_exact}) is at most of order $2^{-r_i}$, one has
\begin{eqnarray}
\Var_{\psi}[g_i(\psi)]
=
O(2^{-r_i}).
\label{eq:BP_scale_result}
\end{eqnarray}
In particular, if $r_i=\Theta(n)$, this is exponentially small in the total system size and thus exhibits barren-plateau scaling.
\end{corollary}

\begin{proof}---The contraction form in Eq.~(\ref{eq:def_local_Haar_var}) is Lemma~\ref{lem:Haar_gradient_benchmark} with $\hB=\hB_i$ and $d=d_i$. To make the decomposition in Eq.~(\ref{eq:var_difference_exact}) explicit, write
\begin{eqnarray}
\hM_{i,\mathrm{LC}}^{(2)}
=
\hM_{\mathrm{Haar},r_i}^{(2)}
+
\left(
\hM_{i,\mathrm{LC}}^{(2)}-
\hM_{\mathrm{Haar},r_i}^{(2)}
\right).
\label{eq:local_moment_Haar_decomposition}
\end{eqnarray}
Substituting this identity into Eq.~(\ref{eq:gradient_variance_local}) and using Eq.~(\ref{eq:def_local_Haar_var}) yields Eq.~(\ref{eq:var_difference_exact}).

For Eq.~(\ref{eq:var_difference_bound}), apply the Hilbert--Schmidt Cauchy--Schwarz inequality:
\begin{eqnarray}
\abs{\Tr\left[ (\hB_i\otimes \hB_i) \left( \hM_{i,\mathrm{LC}}^{(2)}-\hM_{\mathrm{Haar},r_i}^{(2)} \right) \right]} &\le& \norm{\hB_i\otimes \hB_i}_{2}\, \norm{ \hM_{i,\mathrm{LC}}^{(2)} \hM_{\mathrm{Haar},r_i}^{(2)} }_{2} \nonumber\\
	&=& \norm{\hB_i}_{2}^{2}\, \norm{ \hM_{i,\mathrm{LC}}^{(2)}-\hM_{\mathrm{Haar},r_i}^{(2)} }_{2}.
\end{eqnarray}

Finally, Eq.~(\ref{eq:local_Haar_BP_scale}) is Lemma~\ref{lem:Haar_gradient_benchmark} together with Lemma~\ref{lem:Bi_basic}:
\begin{eqnarray}
\Var_{\mathrm{Haar},r_i}[g_i]
\le
\frac{\norm{\hB_i}_{\infty}^{2}}{d_i+1}
\le
\frac{4\norm{\hG_i}_{\infty}^{2}\norm{\hO}_{\infty}^{2}}{2^{r_i}+1}.
\end{eqnarray}
The final statement then follows directly from Eq.~(\ref{eq:var_difference_exact}).
\end{proof}

Corollary~\ref{cor:local2copy_BP} is the precise local statement behind the informal slogan that ``2-design behavior causes the barren plateaus.'' What actually enters the variance of the gradient is the contraction of the local two-copy moment against $\hB_i\otimes \hB_i$. If that contraction is already Haar-like on an effective subsystem of size $r_i$, then the gradient variance is forced down to the Haar scale $O(2^{-r_i})$.

A comment is needed for consistency with Sec.~\ref{sec:setup}. In the strict fixed-circuit/product-input model, Proposition~\ref{prop:spectral_obstruction} prevents literal operator convergence of the full moment orbit to a pure-state Haar moment operator. The same caveat applies here if one interprets $\hM_{i,\mathrm{LC}}^{(2)}$ as arising from a fixed-circuit orbit of local product inputs. Accordingly, the rigorous content of Corollary~\ref{cor:local2copy_BP} should be read at the level of the relevant scalar contraction in Eq.~(\ref{eq:var_difference_exact}), not as an unrestricted operator identity claim. In the usual random-parameter or random-circuit setting of McClean-type theorems, local approximate 2-designity is precisely a sufficient mechanism that makes this contraction Haar-like. In the present fixed-circuit formulation, the same mechanism is isolated as an exact local second-moment criterion.

It is also worth stressing how the light-cone size $r_i$ depends on the architecture and on the choice of cost function. If $\hO$ is a sum of local observables,
\begin{eqnarray}
\hO=\sum_{\alpha} c_{\alpha} \hat{O}_{\alpha},
\label{eq:O_local_sum}
\end{eqnarray}
then
\begin{eqnarray}
\hB_i
=
\sum_{\alpha} c_{\alpha}\, i[\hH_i^{\mathrm{out}},\hat{O}_{\alpha}],
\label{eq:Bi_local_sum}
\end{eqnarray}
and only those terms whose supports overlap the forward cone $\mathcal{F}(i)$ contribute. For bounded-depth circuits on bounded-degree lattices and local costs, the effective light cone often remains small, so the Haar benchmark in Eq.~(\ref{eq:local_Haar_BP_scale}) is not exponentially small in $n$. By contrast, for global costs or sufficiently deep circuits, the effective light cone can grow to $r_i=\Theta(n)$, in which case the same formula yields the familiar exponentially vanishing gradient scale. This recovers, in a single unified expression, the well-known distinction between local-cost and global-cost barren plateaus.

\subsection{Interpretation of global descriptors and local trainability}\label{subsec:BP_interpretation}

We can now articulate the structural relation among the quantities studied so far.

\begin{remark}[Moment hierarchy of coverage, fluctuation, and trainability]\label{rem:moment_hierarchy_BP}
The results of Secs.~\ref{sec:moment} and \ref{sec:gradientBP} can be condensed as
\begin{eqnarray}
\text{BP is local-2-copy},\qquad
\text{EP is global-2-copy-scalar},\qquad
\text{EPD is global-4-copy-scalar}.
\label{eq:BP_EP_EPD_slogan}
\end{eqnarray}
Accordingly, $\mathrm{EP}(\hU)$ alone cannot determine barren-plateau behavior. It is a global mean-entangling descriptor of one fixed circuit, whereas $\Var(g_i)$ is controlled parameter-by-parameter by a local second moment on the effective light cone. The quantity $\mathrm{EPD}(\hU)$ does not determine $\Var(g_i)$ either, but because it probes higher-order fluctuation structure invisible to EP, it can serve as a more informative diagnostic of the randomization process associated with plateau formation.
\end{remark}

Remark~\ref{rem:moment_hierarchy_BP} should be read as the conceptual punchline of the present section. The gradient variance is not governed by a one-dimensional trade-off between ``more entanglement'' and ``less trainability.'' Rather, trainability is a local causal property. The role of global descriptors such as expressibility, EP, and EPD is indirect: they provide increasingly refined information about how broadly and how unevenly the circuit explores the state space, but the actual gradient suppression occurs when the \emph{local} two-copy contraction seen by each parameter has already reached its Haar-like scale.

This is precisely the reason why one can, in principle, have an intermediate regime in which a circuit already has large EP or near-saturated coverage, while the local two-copy moments governing the gradients have not yet fully collapsed. It is also the reason why EPD, being a fluctuation-sensitive descriptor rather than a pure mean descriptor, can empirically align more closely with plateau onset than EP alone. The next sections exploit these observations to develop the structural EP--EPD picture and the associated numerical diagnostics of trainability.

\section{Why Mean Descriptors Are Not Enough}\label{sec:mean_not_enough}

The preceding sections already suggest a strong obstruction to any one-dimensional account of trainability. On the one hand, the fixed-circuit entangling power $\mathrm{EP}(\hU)$ is a single scalar functional of the \emph{global} two-copy output moment:
\begin{eqnarray}
\mathrm{EP}(\hU)=\Tr\left[\hK_Q \hM_{\hU}^{(2)}\right].
\label{eq:mean_not_enough_EP_recall}
\end{eqnarray}
On the other hand, the gradient-variance profile of a parametrized circuit,
\begin{eqnarray}
\Gamma(\hU)
:=
\left(
\Var_{\psi}[g_1(\psi)],
\dots,
\Var_{\psi}[g_p(\psi)]
\right),
\label{eq:def_trainability_profile}
\end{eqnarray}
is determined parameter-by-parameter by a family of \emph{local} two-copy contractions on effective light cones:
\begin{eqnarray}
\Var_{\psi}[g_i(\psi)]
=
\Tr\left[
(\hB_i\otimes \hB_i)\,
\hM_{i,\mathrm{LC}}^{(2)}
\right].
\label{eq:mean_not_enough_var_recall}
\end{eqnarray}
The structural mismatch is immediate: \emph{a single global scalar cannot, in general, encode an entire parameter-indexed collection of local second moments.}

The aim of the present section is to make this obstruction explicit. We do so in the strongest possible form, namely by constructing two entangling blocks that have the \emph{same mean entangling power} but induce \emph{different gradient variances} when inserted into the same trainable architecture. In this sense, the failure of mean descriptors is not merely a matter of asymptotic scaling or numerical correlation. It is already present at the level of an explicit, finite-dimensional, analytically tractable counterexample.

At the same time, the example will show why EPD is a useful refinement. Although EPD is still not the trainability object itself, it distinguishes the two blocks that EP cannot distinguish. This is exactly the kind of separation one expects from the hierarchy established earlier:
\begin{eqnarray}
\text{EP: global mean,}
\qquad
\text{EPD: global fluctuation,}
\qquad
\text{BP: local second-moment trainability.}
\label{eq:hierarchy_recall_mean_not_enough}
\end{eqnarray}

\subsection{A one-dimensional descriptor cannot parametrize a trainability profile}\label{subsec:nonidentifiability_setup}

Before giving the explicit construction, it is useful to state precisely what is meant by ``mean descriptors are not enough.'' For a fixed architecture and fixed cost class, the trainability data of interest is not a single number but the vector $\Gamma(\hU)$ in Eq.~(\ref{eq:def_trainability_profile}). Even if one focuses on a single designated parameter, the associated gradient variance remains a local object that depends on the interaction between the parameter generator, the cost observable, and the effective light cone. This already suggests non-identifiability from any one-number global mean descriptor.

To prepare the constructive theorem below, we record two elementary lemmas.

\begin{lemma}[Conjugation of anticommuting Pauli strings]\label{lem:anticommuting_pauli_conjugation}
Let $\hat P$ and $\hat Q$ be Hermitian Pauli strings satisfying
\begin{eqnarray}
\hat P^2=\hat Q^2=\hid,
\qquad
\{\hat P,\hat Q\}=0.
\end{eqnarray}
Then, for all real $\theta$,
\begin{eqnarray}
e^{i\theta \hat P} \hat Q e^{-i\theta \hat P}
=
\hat Q\cos(2\theta)
+
i\hat P\hat Q \sin(2\theta).
\label{eq:anticommuting_pauli_conjugation}
\end{eqnarray}
\end{lemma}

\begin{proof}---Define
\begin{eqnarray}
F(\theta):=e^{i\theta \hat P}\hat Qe^{-i\theta \hat P}.
\end{eqnarray}
Then,
\begin{eqnarray}
F'(\theta) = ie^{i\theta \hat P}[\hat P,\hat Q]e^{-i\theta \hat P} = 2i\,e^{i\theta \hat P}\hat P\hat Q\,e^{-i\theta \hat P},
\label{eq:Fprime_anticomm}
\end{eqnarray}
because $\{\hat P,\hat Q\}=0$ implies $[\hat P,\hat Q]=2\hat P\hat Q$. Likewise, with $R(\theta):=e^{i\theta \hat P}(i\hat P\hat Q)e^{-i\theta \hat P}$, one finds
\begin{eqnarray}
R'(\theta)
&=&
ie^{i\theta \hat P}[\hat P,i\hat P\hat Q]e^{-i\theta \hat P}.
\end{eqnarray}
Since $\hat P(i\hat P\hat Q)=i\hat Q$ and $(i\hat P\hat Q)\hat P=-i\hat Q$, we have
\begin{eqnarray}
[\hat P,i\hat P\hat Q]=2i\hat Q,
\end{eqnarray}
and hence,
\begin{eqnarray}
R'(\theta)=-2F(\theta).
\label{eq:Rprime_anticomm}
\end{eqnarray}
The pair $(F,R)$ therefore satisfies the linear system
\begin{eqnarray}
F'(\theta)=2R(\theta),
\qquad
R'(\theta)=-2F(\theta),
\end{eqnarray}
with initial conditions
\begin{eqnarray}
F(0)=\hat Q,
\qquad
R(0)=i\hat P\hat Q.
\end{eqnarray}
The unique solution is exactly Eq.~(\ref{eq:anticommuting_pauli_conjugation}).
\end{proof}

\begin{lemma}[Single-qubit Haar moments of Bloch coordinates]\label{lem:single_qubit_Haar_bloch}
Let $\ket{\psi}\sim \mu_{\Haar}^{(1)}$ be Haar-random on one qubit, and define
\begin{eqnarray}
x:=\bra{\psi}\hX\ket{\psi},
\qquad
y:=\bra{\psi}\hY\ket{\psi},
\qquad
z:=\bra{\psi}\hZ\ket{\psi}.
\end{eqnarray}
Then,
\begin{eqnarray}
\mathbb{E}[x]=\mathbb{E}[y]=\mathbb{E}[z]=0,
\quad\text{and}\quad
\mathbb{E}[x^2]=\mathbb{E}[y^2]=\mathbb{E}[z^2]=\frac{1}{3}.
\label{eq:bloch_second_moments}
\end{eqnarray}
More generally,
\begin{eqnarray}
\mathbb{E}[x_ax_b]=\frac{\delta_{ab}}{3},
\qquad
a,b\in\{x,y,z\}.
\label{eq:bloch_covariance}
\end{eqnarray}
\end{lemma}

\begin{proof}---The Haar measure on one qubit induces the uniform distribution on the Bloch sphere. Hence the random vector
\begin{eqnarray}
\vec r := (x,y,z)
\end{eqnarray}
is rotationally invariant and satisfies
\begin{eqnarray}
x^2+y^2+z^2=1
\end{eqnarray}
almost surely. Rotational invariance implies
\begin{eqnarray}
\mathbb{E}[x]=\mathbb{E}[y]=\mathbb{E}[z]=0
\end{eqnarray}
and
\begin{eqnarray}
\mathbb{E}[x^2]=\mathbb{E}[y^2]=\mathbb{E}[z^2].
\end{eqnarray}
Summing these equal quantities and using $\mathbb{E}[x^2+y^2+z^2]=1$ gives Eq.~(\ref{eq:bloch_second_moments}). The mixed moments in Eq.~(\ref{eq:bloch_covariance}) vanish by rotational symmetry.
\end{proof}

\subsection{Constructive equal-EP, unequal-trainability theorem}\label{subsec:constructive_separation}

We now present the promised explicit separation. The two entanglers are chosen from the canonical two-qubit KAK family analyzed in the preceding EPD work~\cite{SM-ChoBang2026EPD}. Concretely, define
\begin{eqnarray}
\hE_{\mathrm{C}}
:=
\exp\left(
-i\frac{\pi}{4}\hX_1\hX_2
\right),
\label{eq:def_EC}
\end{eqnarray}
and
\begin{eqnarray}
\hE_{\mathrm{B}}
:=
\exp\left(
-i\frac{\pi}{4}\hX_1\hX_2
-i\frac{\pi}{8}\hY_1\hY_2
\right).
\label{eq:def_EB}
\end{eqnarray}
Because $\hX_1\hX_2$ and $\hY_1\hY_2$ commute, one may also write
\begin{eqnarray}
\hE_{\mathrm{B}}
=
\exp\left(
-i\frac{\pi}{4}\hX_1\hX_2
\right)
\exp\left(
-i\frac{\pi}{8}\hY_1\hY_2
\right).
\label{eq:def_EB_factorized}
\end{eqnarray}
These are the canonical Cartan representatives with parameters
\begin{eqnarray}
\chi_{\mathrm{C}}=\left(\frac{\pi}{4},0,0\right),
\qquad
\chi_{\mathrm{B}}=\left(\frac{\pi}{4},\frac{\pi}{8},0\right),
\end{eqnarray}
which are locally equivalent to the CNOT and $B$-gate classes, respectively~\cite{SM-ChoBang2026EPD}.

We embed each entangler into the same one-parameter hardware-efficient block,
\begin{eqnarray}
\hU_{\alpha}(\theta)
:=
\hE_{\alpha}
\left(
e^{-i\theta \hZ_1/2}\otimes \hid_2
\right),
\qquad
\alpha\in\{\mathrm{C},\mathrm{B}\},
\label{eq:def_UEalpha}
\end{eqnarray}
and consider the same cost observable for both:
\begin{eqnarray}
\hO := \hid_1\otimes \hZ_2.
\label{eq:def_O_separation}
\end{eqnarray}
For an input product state $\ket{\psi}=\ket{\psi_1}\otimes \ket{\psi_2}$ sampled from $\mu_{\mathrm{prod}}$, define
\begin{eqnarray}
C_{\alpha}(\theta,\psi)
:=
\bra{\psi}
\hU_{\alpha}(\theta)^{\dagger}
\hO
\hU_{\alpha}(\theta)
\ket{\psi},
\label{eq:def_Calpha}
\end{eqnarray}
and the corresponding gradient
\begin{eqnarray}
g_{\alpha}(\psi,\theta)
:=
\partial_{\theta} C_{\alpha}(\theta,\psi).
\label{eq:def_galpha}
\end{eqnarray}

We can now state the central result of this section.

\begin{theorem}[Constructive equal-EP, unequal-trainability theorem]\label{thm:constructive_equalEP_unequaltrainability}
The two one-parameter blocks $\hU_{\mathrm{C}}(\theta)$ and $\hU_{\mathrm{B}}(\theta)$ defined in Eq.~(\ref{eq:def_UEalpha}) satisfy the following properties.

\textbf{(i) Equal mean entangling power but unequal entangling-power deviation~\cite{SM-ChoBang2026EPD}.}
For every $\theta$,
\begin{eqnarray}
\mathrm{EP}(\hU_{\mathrm{C}}(\theta))
=
\mathrm{EP}(\hU_{\mathrm{B}}(\theta))
=
\frac{4}{9},
\label{eq:equal_EP_theta_all}
\end{eqnarray}
whereas
\begin{eqnarray}
\mathrm{EPD}(\hU_{\mathrm{C}}(\theta))
=
\frac{4\sqrt{11}}{45},
\qquad
\mathrm{EPD}(\hU_{\mathrm{B}}(\theta))
=
\frac{2\sqrt{35}}{45}.
\label{eq:unequal_EPD_theta_all}
\end{eqnarray}
In particular,
\begin{eqnarray}
\mathrm{EP}(\hU_{\mathrm{C}}(\theta))
=
\mathrm{EP}(\hU_{\mathrm{B}}(\theta)),
\qquad
\mathrm{EPD}(\hU_{\mathrm{C}}(\theta))
\neq
\mathrm{EPD}(\hU_{\mathrm{B}}(\theta)).
\label{eq:sameEP_diffEPD}
\end{eqnarray}

\textbf{(ii) Unequal gradient variance in the same trainable architecture.}
For every $\theta$,
\begin{eqnarray}
\Var_{\psi}\left[g_{\mathrm{C}}(\psi,\theta)\right]
=
\frac{1}{9},
\qquad
\Var_{\psi}\left[g_{\mathrm{B}}(\psi,\theta)\right]
=
\frac{1}{18}.
\label{eq:grad_var_explicit_difference}
\end{eqnarray}
Hence, the two blocks have identical mean entangling power but different trainability, even though the trainable layer, parameter placement, and cost observable are otherwise the same.
\end{theorem}

\begin{proof}---We divide the proof into two parts.

\paragraph*{Part I: same EP, different EPD.}
By Remark~\ref{rem:LU_invariance}, EP and EPD are invariant under local pre- and post-unitaries. Therefore, the fixed-circuit descriptors of $\hE_{\mathrm{C}}$ and $\hE_{\mathrm{B}}$ agree with those of the CNOT and $B$ local-equivalence classes analyzed in Ref.~\cite{SM-ChoBang2026EPD}. That work uses the two-qubit linear entropy
\begin{eqnarray}
E(\ket{\phi}) = 1-\Tr \hrho_1^2,
\end{eqnarray}
whereas our present two-qubit descriptor is the Meyer--Wallach quantity
\begin{eqnarray}
Q(\ket{\phi})
=
\frac{2}{2}\sum_{j=1}^{2}(1-\Tr \hrho_j^2)
=
2E(\ket{\phi}).
\label{eq:Q_equals_2E_twobits}
\end{eqnarray}
Hence, the EP and EPD values from Ref.~\cite{SM-ChoBang2026EPD} are multiplied by a factor of $2$ in our normalization. Since the CNOT and $B$ classes have
\begin{eqnarray}
\mathrm{ep}_{E}(\mathrm{CNOT})
=
\mathrm{ep}_{E}(B)
=
\frac{2}{9},
\qquad
\Delta_{E}(\mathrm{CNOT})
=
\frac{2\sqrt{11}}{45},
\qquad
\Delta_{E}(B)
=
\frac{\sqrt{35}}{45},
\end{eqnarray}
we obtain
\begin{eqnarray}
\mathrm{EP}(\hE_{\mathrm{C}})
=
\mathrm{EP}(\hE_{\mathrm{B}})
=
\frac{4}{9},
\label{eq:EP_EC_EB}
\end{eqnarray}
and
\begin{eqnarray}
\mathrm{EPD}(\hE_{\mathrm{C}})
=
\frac{4\sqrt{11}}{45},
\qquad
\mathrm{EPD}(\hE_{\mathrm{B}})
=
\frac{2\sqrt{35}}{45}.
\label{eq:EPD_EC_EB}
\end{eqnarray}
Now $\hU_{\alpha}(\theta)$ differs from $\hE_{\alpha}$ only by right multiplication with the local unitary $e^{-i\theta \hZ_1/2}\otimes \hid_2$, so Remark~\ref{rem:LU_invariance} gives
\begin{eqnarray}
\mathrm{EP}(\hU_{\alpha}(\theta))=\mathrm{EP}(\hE_{\alpha}),
\qquad
\mathrm{EPD}(\hU_{\alpha}(\theta))=\mathrm{EPD}(\hE_{\alpha})
\end{eqnarray}
for all $\theta$ and both $\alpha\in\{\mathrm{C},\mathrm{B}\}$. This proves Eqs.~(\ref{eq:equal_EP_theta_all})--(\ref{eq:sameEP_diffEPD}).

\paragraph*{Part II: explicit gradient variance calculation.}
We now compute the gradient with respect to the trainable $\hZ_1$-rotation. Write
\begin{eqnarray}
\hR_Z(\theta):=e^{-i\theta \hZ_1/2}\otimes \hid_2.
\end{eqnarray}
Then,
\begin{eqnarray}
\hU_{\alpha}(\theta)=\hE_{\alpha}\hR_Z(\theta),
\end{eqnarray}
so Eq.~(\ref{eq:def_Calpha}) becomes
\begin{eqnarray}
C_{\alpha}(\theta,\psi)
=
\bra{\psi}
\hR_Z(\theta)^{\dagger}
\hE_{\alpha}^{\dagger}\hO\hE_{\alpha}
\hR_Z(\theta)
\ket{\psi}.
\label{eq:Calpha_input_picture}
\end{eqnarray}
Differentiating,
\begin{eqnarray}
g_{\alpha}(\psi,\theta)
&=&
\bra{\psi}
\hR_Z(\theta)^{\dagger}
\hA_{\alpha}
\hR_Z(\theta)
\ket{\psi},
\label{eq:galpha_input_picture}
\end{eqnarray}
where
\begin{eqnarray}
\hA_{\alpha}
:=
\frac{i}{2}
\left[
\hZ_1\otimes \hid_2,
\hE_{\alpha}^{\dagger}\hO\hE_{\alpha}
\right].
\label{eq:def_Aalpha}
\end{eqnarray}
Thus, the problem reduces to computing the dressed observable $\hE_{\alpha}^{\dagger}\hO\hE_{\alpha}$.

We first treat the $\mathrm{C}$ block. Since $\hX_1\hX_2$ anticommutes with $\hO=\hid_1\otimes \hZ_2$, Lemma~\ref{lem:anticommuting_pauli_conjugation} with
\begin{eqnarray}
\hat P=\hX_1\hX_2,
\qquad
\hat Q=\hid_1\otimes \hZ_2,
\qquad
\theta=\frac{\pi}{4}
\end{eqnarray}
yields
\begin{eqnarray}
\hE_{\mathrm{C}}^{\dagger}\hO\hE_{\mathrm{C}}
=
e^{i\frac{\pi}{4}\hX_1\hX_2}
(\hid_1\otimes \hZ_2)
e^{-i\frac{\pi}{4}\hX_1\hX_2}
=
i(\hX_1\hX_2)(\hid_1\otimes \hZ_2).
\label{eq:EC_conjugates_O}
\end{eqnarray}
Using $\hX\hZ=-i\hY$, we obtain
\begin{eqnarray}
i(\hX_1\hX_2)(\hid_1\otimes \hZ_2)
=
i\left(\hX_1\otimes \hX_2\hZ_2\right)
=
i\left(\hX_1\otimes (-i\hY_2)\right)
=
\hX_1\hY_2.
\label{eq:EC_conjugates_O_XY}
\end{eqnarray}
Therefore,
\begin{eqnarray}
\hA_{\mathrm{C}}
&=&
\frac{i}{2}
\left[
\hZ_1\otimes \hid_2,
\hX_1\hY_2
\right].
\end{eqnarray}
Now
\begin{eqnarray}
(\hZ_1\otimes \hid_2)(\hX_1\hY_2)=i\hY_1\hY_2,
\qquad
(\hX_1\hY_2)(\hZ_1\otimes \hid_2)=-i\hY_1\hY_2,
\end{eqnarray}
so
\begin{eqnarray}
\hA_{\mathrm{C}}
=
\frac{i}{2}(2i\hY_1\hY_2)
=
-\,\hY_1\hY_2.
\label{eq:AC_equals_YY}
\end{eqnarray}

Next consider the $\mathrm{B}$ block. Since $\hX_1\hX_2$ and $\hY_1\hY_2$ commute, Eq.~(\ref{eq:def_EB_factorized}) gives
\begin{eqnarray}
\hE_{\mathrm{B}}^{\dagger}\hO\hE_{\mathrm{B}}
&=&
e^{i\frac{\pi}{8}\hY_1\hY_2}
\left(
e^{i\frac{\pi}{4}\hX_1\hX_2}
(\hid_1\otimes \hZ_2)
e^{-i\frac{\pi}{4}\hX_1\hX_2}
\right)
e^{-i\frac{\pi}{8}\hY_1\hY_2}
\nonumber\\
&=&
e^{i\frac{\pi}{8}\hY_1\hY_2}
(\hX_1\hY_2)
e^{-i\frac{\pi}{8}\hY_1\hY_2},
\label{eq:EB_conjugation_step1}
\end{eqnarray}
where Eq.~(\ref{eq:EC_conjugates_O_XY}) was used. Since $\hY_1\hY_2$ anticommutes with $\hX_1\hY_2$, Lemma~\ref{lem:anticommuting_pauli_conjugation} now gives
\begin{eqnarray}
\hE_{\mathrm{B}}^{\dagger}\hO\hE_{\mathrm{B}}
&=&
(\hX_1\hY_2)\cos\frac{\pi}{4}
+
i(\hY_1\hY_2)(\hX_1\hY_2)\sin\frac{\pi}{4}.
\end{eqnarray}
Because $\hY_1\hX_1=-i\hZ_1$ and $\hY_2\hY_2=\hid_2$, we have
\begin{eqnarray}
i(\hY_1\hY_2)(\hX_1\hY_2)
=
i(\hY_1\hX_1\otimes \hid_2)
=
i(-i\hZ_1\otimes \hid_2)
=
\hZ_1\otimes \hid_2.
\end{eqnarray}
Hence
\begin{eqnarray}
\hE_{\mathrm{B}}^{\dagger}\hO\hE_{\mathrm{B}}
=
\frac{1}{\sqrt{2}}
\left(
\hX_1\hY_2 + \hZ_1\otimes \hid_2
\right).
\label{eq:EB_conjugates_O}
\end{eqnarray}
Substituting into Eq.~(\ref{eq:def_Aalpha}),
\begin{eqnarray}
\hA_{\mathrm{B}} = \frac{i}{2} \left[ \hZ_1\otimes \hid_2, \frac{1}{\sqrt{2}} \left( \hX_1\hY_2+\hZ_1\otimes \hid_2 \right) \right] = \frac{i}{2\sqrt{2}} \left[ \hZ_1\otimes \hid_2, \hX_1\hY_2 \right] = -\frac{1}{\sqrt{2}}\,\hY_1\hY_2,
\label{eq:AB_equals_YY}
\end{eqnarray}
where Eq.~(\ref{eq:AC_equals_YY}) was used in the last step.

Eqs.~(\ref{eq:AC_equals_YY}) and (\ref{eq:AB_equals_YY}) show that the two architectures reduce to the same input observable $\hY_1\hY_2$, but with different amplitudes. For an input product state $\ket{\psi_1}\otimes \ket{\psi_2}$, define
\begin{eqnarray}
y_j := \bra{\psi_j}\hY_j\ket{\psi_j},
\qquad
j=1,2.
\end{eqnarray}
Then Eq.~(\ref{eq:galpha_input_picture}) gives
\begin{eqnarray}
g_{\mathrm{C}}(\psi,\theta)
=
-\,
\bra{\psi}
\hR_Z(\theta)^{\dagger}(\hY_1\hY_2)\hR_Z(\theta)
\ket{\psi},
\label{eq:gC_theta}
\end{eqnarray}
and
\begin{eqnarray}
g_{\mathrm{B}}(\psi,\theta)
=
-\frac{1}{\sqrt{2}}\,
\bra{\psi}
\hR_Z(\theta)^{\dagger}(\hY_1\hY_2)\hR_Z(\theta)
\ket{\psi}.
\label{eq:gB_theta}
\end{eqnarray}
Now $\hR_Z(\theta)$ is a local unitary and $\mu_{\mathrm{prod}}$ is invariant under independent single-qubit Haar rotations. Therefore the distribution of
\begin{eqnarray}
\bra{\psi}
\hR_Z(\theta)^{\dagger}(\hY_1\hY_2)\hR_Z(\theta)
\ket{\psi}
\end{eqnarray}
is the same as that of
\begin{eqnarray}
\bra{\psi}\hY_1\hY_2\ket{\psi}
=
y_1y_2.
\end{eqnarray}
Hence the gradient distribution is independent of $\theta$, and it suffices to compute at $\theta=0$:
\begin{eqnarray}
g_{\mathrm{C}}(\psi,0)=-y_1y_2,
\qquad
g_{\mathrm{B}}(\psi,0)=-\frac{1}{\sqrt{2}}y_1y_2.
\label{eq:gC_gB_y1y2}
\end{eqnarray}

By Lemma~\ref{lem:single_qubit_Haar_bloch}, the independent random variables $y_1$ and $y_2$ satisfy
\begin{eqnarray}
\mathbb{E}[y_j]=0,
\qquad
\mathbb{E}[y_j^2]=\frac{1}{3}.
\end{eqnarray}
Therefore,
\begin{eqnarray}
\mathbb{E}[y_1y_2]=\mathbb{E}[y_1]\mathbb{E}[y_2]=0,
\end{eqnarray}
and
\begin{eqnarray}
\mathbb{E}[y_1^2y_2^2]
=
\mathbb{E}[y_1^2]\mathbb{E}[y_2^2]
=
\frac{1}{9}.
\end{eqnarray}
Thus
\begin{eqnarray}
\Var_{\psi}[g_{\mathrm{C}}(\psi,\theta)]
&=&
\Var_{\psi}[y_1y_2]
=
\frac{1}{9},
\\
\Var_{\psi}[g_{\mathrm{B}}(\psi,\theta)]
&=&
\frac{1}{2}\Var_{\psi}[y_1y_2]
=
\frac{1}{18},
\end{eqnarray}
for all $\theta$. This proves Eq.~(\ref{eq:grad_var_explicit_difference}).
\end{proof}

Theorem~\ref{thm:constructive_equalEP_unequaltrainability} is the key structural separation result of the paper. It exhibits, in a fully analytic manner, two circuit blocks that are indistinguishable by the scalar mean descriptor EP but distinguishable both by EPD and by an actual trainability metric, namely the gradient variance of the same parameter in the same trainable architecture.

Several important consequences follow immediately.

\subsection{Non-identifiability of barren plateaus from scalar mean descriptors}\label{subsec:nonidentifiability_formal}

We can now state the general structural conclusion.

\begin{proposition}[Non-identifiability of trainability from scalar mean descriptors]\label{prop:nonidentifiability_EP}
Fix an architecture class and a cost observable, and let
\begin{eqnarray}
\Gamma(\hU)
=
\left(
\Var_{\psi}[g_1(\psi)],
\dots,
\Var_{\psi}[g_p(\psi)]
\right)
\end{eqnarray}
denote the gradient-variance profile with respect to the Haar-product input ensemble. Then the map $\hU\mapsto \Gamma(\hU)$ does \emph{not} factor through the scalar mean descriptor $\mathrm{EP}(\hU)$. Equivalently, there does not exist a function
\begin{eqnarray}
F:\mathbb{R}\to \mathbb{R}^{p}
\end{eqnarray}
such that
\begin{eqnarray}
\Gamma(\hU)=F(\mathrm{EP}(\hU))
\qquad
\text{for all }\hU
\end{eqnarray}
in the architecture class.
\end{proposition}

\begin{proof}
Theorem~\ref{thm:constructive_equalEP_unequaltrainability} provides two one-parameter blocks $\hU_{\mathrm{C}}(\theta)$ and $\hU_{\mathrm{B}}(\theta)$ such that
\begin{eqnarray}
\mathrm{EP}(\hU_{\mathrm{C}}(\theta))
=
\mathrm{EP}(\hU_{\mathrm{B}}(\theta))
\end{eqnarray}
for all $\theta$, while
\begin{eqnarray}
\Var_{\psi}[g_{\mathrm{C}}(\psi,\theta)]
\neq
\Var_{\psi}[g_{\mathrm{B}}(\psi,\theta)].
\end{eqnarray}
Hence even the variance of a \emph{single} designated parameter cannot be recovered as a function of $\mathrm{EP}(\hU)$ alone. A fortiori, the full profile $\Gamma(\hU)$ cannot factor through $\mathrm{EP}(\hU)$.
\end{proof}

Proposition~\ref{prop:nonidentifiability_EP} is stronger than a mere empirical warning. It says that the insufficiency of mean descriptors is a matter of principle, not of imperfect regression or finite-sample error. The map
\begin{eqnarray}
\hU \longmapsto \mathrm{EP}(\hU)
\end{eqnarray}
forgets information that is genuinely relevant to trainability.

This conclusion is entirely consistent with the moment-theoretic analysis of Secs.~\ref{sec:moment} and \ref{sec:gradientBP}. The quantity $\mathrm{EP}(\hU)$ is a single contraction of the global two-copy output moment against the observable $\hK_Q$. By contrast, each gradient variance is a different contraction of a parameter-dependent local two-copy moment against a parameter-dependent local observable $\hB_i\otimes \hB_i$. There is no reason, in general, for all of that local information to be recoverable from one global scalar.

The explicit example above also clarifies the role of EPD. Since
\begin{eqnarray}
\mathrm{EP}(\hU_{\mathrm{C}}(\theta))
=
\mathrm{EP}(\hU_{\mathrm{B}}(\theta)),
\qquad
\mathrm{EPD}(\hU_{\mathrm{C}}(\theta))
\neq
\mathrm{EPD}(\hU_{\mathrm{B}}(\theta)),
\end{eqnarray}
the higher-order fluctuation descriptor already distinguishes the two blocks before any gradient is computed. In this sense, EPD is strictly more informative than EP for this pair. However, it is important not to overstate the conclusion.

\begin{remark}[EPD is more informative than EP, but still not the trainability invariant]\label{rem:EPD_not_complete}
Theorem~\ref{thm:constructive_equalEP_unequaltrainability} does \emph{not} imply that EPD alone determines barren-plateau behavior. EPD remains a global scalar four-copy quantity, whereas gradient variance is a local two-copy quantity. What the theorem does show is more precise: EPD can distinguish circuits that are invisible to EP and that, once embedded into the same trainable block, exhibit different gradient variances. Thus EPD is not the trainability object itself, but it is a strictly richer diagnostic than EP.
\end{remark}

The non-identifiability statement can be rephrased in the language of information loss. A scalar mean descriptor compresses the circuit to one number, the trainability profile is, in general, a many-parameter object. Theorem~\ref{thm:constructive_equalEP_unequaltrainability} shows that this compression is already too coarse at the two-qubit block level. Consequently, when such blocks are used as building units of hardware-efficient ansätze, the inability of a mean descriptor to resolve local trainability is not an exceptional pathology but a generic structural limitation.

\subsection{Interpretation for expressibility and barren-plateau phenomenology}\label{subsec:mean_not_enough_interpretation}

Although the explicit construction above is formulated in terms of fixed-circuit EP and EPD, its interpretive consequence extends directly to the broader expressibility literature. Indeed, scalar expressibility indices are likewise one-number summaries of a circuit family. Such descriptors can correlate with broad architectural trends, and in many random-circuit settings they do. But the present analysis shows why no such scalar should be expected to provide a complete criterion for plateau onset.

The correct hierarchy is now clear:
\begin{enumerate}
\item A \emph{global mean descriptor} such as EP measures average entangling strength.
\item A \emph{global fluctuation descriptor} such as EPD measures the input dependence of that strength.
\item The \emph{actual trainability object} is the parameter-resolved local two-copy contraction on the effective light cone.
\end{enumerate}
Accordingly, the common informal slogan
\begin{eqnarray}
\text{``more expressibility''} \Longrightarrow \text{``more barren plateau''}
\end{eqnarray}
must be interpreted with care. What can be true, under additional assumptions, is that stronger randomization tends to push several descriptors in the same direction. But the descriptors live at different levels of the moment hierarchy and encode different information. A mean descriptor can saturate before the local second moments relevant to trainability have fully collapsed.

This is exactly the regime that motivates the practical two-dial picture advocated in this paper. One dial measures coverage or average entangling strength, the other measures fluctuation or variability. Theorem~\ref{thm:constructive_equalEP_unequaltrainability} shows that fixing the first dial does not fix the second, nor does it fix trainability. Therefore any meaningful account of barren plateaus must go beyond a single mean descriptor.

We summarize the lesson of this section as follows:
\begin{quote}
\centering
{\em Equal mean entangling power does not imply equal trainability.}
\end{quote}
The obstruction is not numerical but structural, and it arises already at the smallest nontrivial entangling scale. This is why the transition from mean descriptors to fluctuation-sensitive and local-moment-sensitive diagnostics is not merely a refinement of language, but a necessary step in understanding PQC trainability.

\section{Numerical protocols and supplementary diagnostics}

This section gives the numerical details supporting the depth-onset and route-scan claims in the main manuscript. The guiding question is the following: as the depth $L$ of a PQC family is increased, can one locate a regime in which the circuit has already acquired substantial coverage while its product-input response has not yet homogenized?

All simulations use product-Haar input probes for fixed-circuit EP and EPD. Family-level quantities are obtained in fixed-circuit order: for each sampled parameter vector, the product-input distribution of $Q(\hU_L(\bm{\theta})\ket{\psi})$ is first estimated; the corresponding fixed-circuit mean and standard deviation give EP and EPD; only after this step are the resulting descriptors averaged over parameter samples. This ordering is essential. It keeps the present EP/EPD descriptors distinct from Sim-style parameter-ensemble expressibility, in which one samples many circuit instances from a fixed reference input.

The numerical section should therefore be read with the same hierarchy as the analytical sections:
\begin{eqnarray}
\text{coverage dial} &:& \text{how far the family has moved toward high mean entanglement},\nonumber\\
\text{variability dial} &:& \text{how much product-input response width remains},\nonumber\\
\text{gradient probe} &:& \text{a cost-dependent consistency check of local trainability}.
\label{eq:sm_numeric_hierarchy}
\end{eqnarray}
The first two entries define the EP/EPD scan; the third is not used to define the scan and is included only to compare the diagnostic with an explicit gradient scale.

\subsection{Main example: Ring-CZ depth scan}

The ring-CZ hardware-efficient ansatz used in the onset diagnostic is
\begin{eqnarray}
\hU_L(\bm{\theta})=\prod_{\ell=1}^{L}
\left[
\left(\prod_{j=1}^{n}CZ_{j,j+1}\right)
\left(\prod_{j=1}^{n}\hR_z(\phi_{\ell j})\hR_y(\vartheta_{\ell j})\right)
\right],
\label{eq:sm_ringcz}
\end{eqnarray}
with periodic boundary conditions. All parameters are sampled independently from $[0,2\pi)$. This family is useful as a first analysis example because it exhibits the separation between three depth-dependent events: coverage saturation, mean-EP growth, and variability/gradient collapse.

The Sim-style coverage proxy is~\cite{SM-Sim2019expressibility}
\begin{eqnarray}
X_L=\exp[-D_{\rm KL}(P_L(F)\Vert P_{\Haar}(F))],
\label{eq:sm_expr_proxy}
\end{eqnarray}
where $P_L(F)$ is the empirical pairwise-fidelity distribution obtained from parameter sampling and the reference input $\ket{0}^{\otimes n}$. A larger $X_L$ means that the parameter-sampled fidelity distribution is closer to the Haar fidelity distribution. This quantity is intentionally included as a contrast: it is a parameter-ensemble coverage proxy, whereas EP and EPD are fixed-circuit product-input descriptors.

For the onset comparison, we normalize increasing quantities as
\begin{eqnarray}
\widetilde X_L=\frac{X_L}{\max_{L'}X_{L'}},
\qquad
\widetilde{\EP}_L=\frac{\EP_L-\min_{L'}\EP_{L'}}{\max_{L'}\EP_{L'}-\min_{L'}\EP_{L'}},
\label{eq:sm_normalized_increasing}
\end{eqnarray}
and define collapse ratios
\begin{eqnarray}
R_{\EPD}(L)=\frac{\EPD_L}{\EPD_1},
\qquad
R_g(L)=\frac{\overline V_L}{\overline V_1}.
\label{eq:sm_collapse_ratios}
\end{eqnarray}
Here, $\overline V_L$ denotes the input-induced variance of a representative finite-difference gradient at depth $L$. Thus, $R_{\EPD}$ and $R_g$ decrease when the output-entanglement distribution and the representative gradient scale collapse relative to their $L=1$ values.

\begin{figure}[t]
\centering
\includegraphics[width=0.88\textwidth]{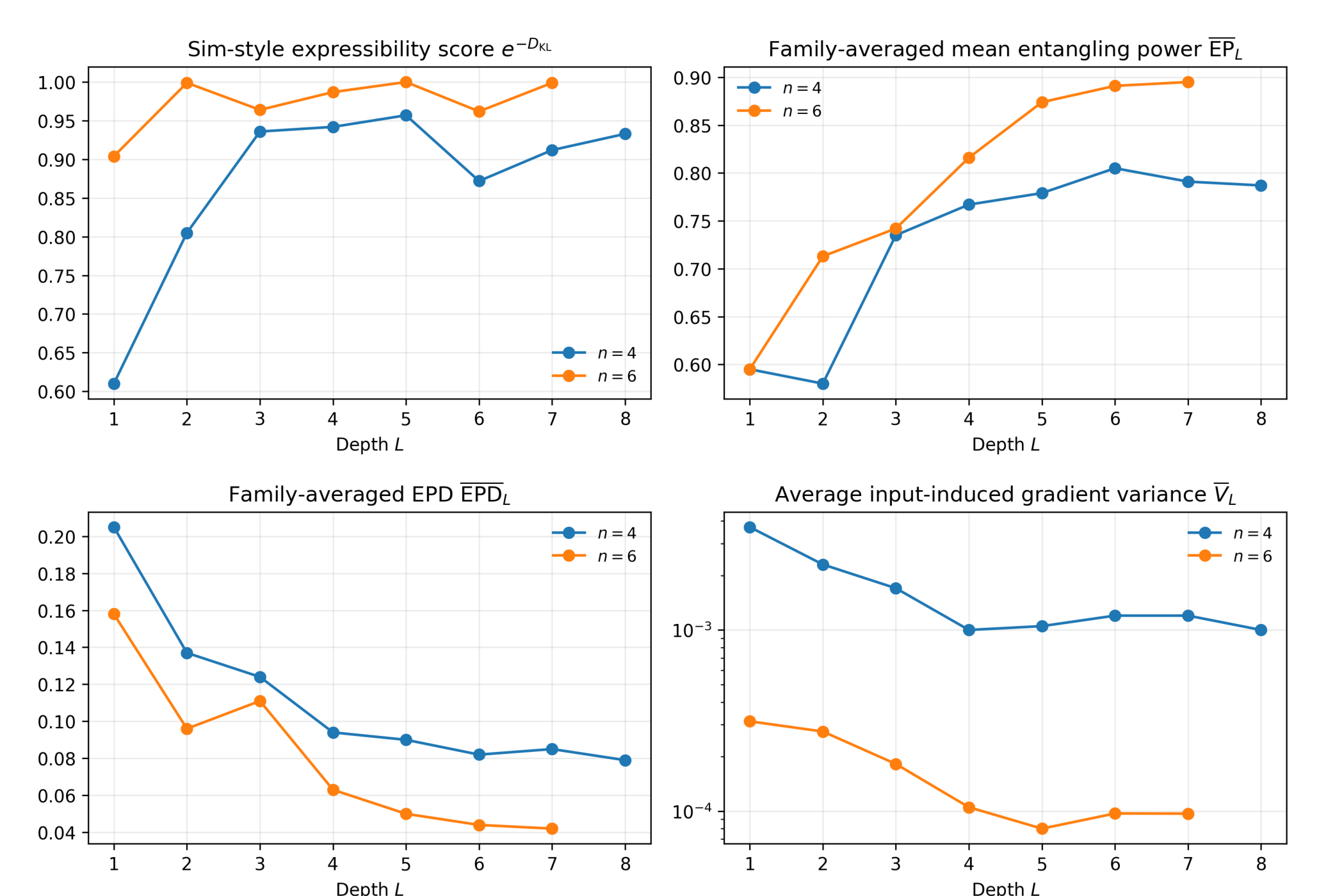}
\caption{Depth scan for the ring-CZ hardware-efficient ansatz. The four panels show the Sim-style coverage proxy, family-averaged EP, family-averaged EPD, and the input-induced variance of a representative gradient. The key feature is the mismatch between the early saturation of the coverage proxy and the later collapse of EPD and the gradient scale.}
\label{fig:sm_depth}
\end{figure}

Fig.~\ref{fig:sm_depth} should be read as follows: (a) The coverage proxy $X_L$ rises rapidly with depth. For the finite-size data shown here, the $n=6$ curve is already close to its scan maximum at small depth. This means that the parameter ensemble has quickly learned to reproduce a Haar-like fidelity distribution. (b) The mean fixed-circuit EP generally increases with depth, but it does not carry the same information as the coverage proxy. It measures the average Meyer-Wallach entanglement generated from product inputs by a fixed circuit and then averaged over circuit instances. (c) EPD decreases after an intermediate depth. This is the narrowing of the product-input output-entanglement distribution: different product inputs begin to produce nearly the same amount of entanglement. (d) The representative gradient variance also decreases, on a logarithmic scale, in the same depth range where EPD collapses. This is consistent with the local two-copy criterion of Sec.~III, although it is not a universal functional relation because the gradient variance depends on the chosen parameter and cost. Thus, Fig.~\ref{fig:sm_depth} separates two statements that are often conflated: a circuit family can look highly covering according to a fidelity-distribution proxy before the fixed-circuit response has fully homogenized.

\begin{figure}[t]
\centering
\includegraphics[width=0.80\textwidth]{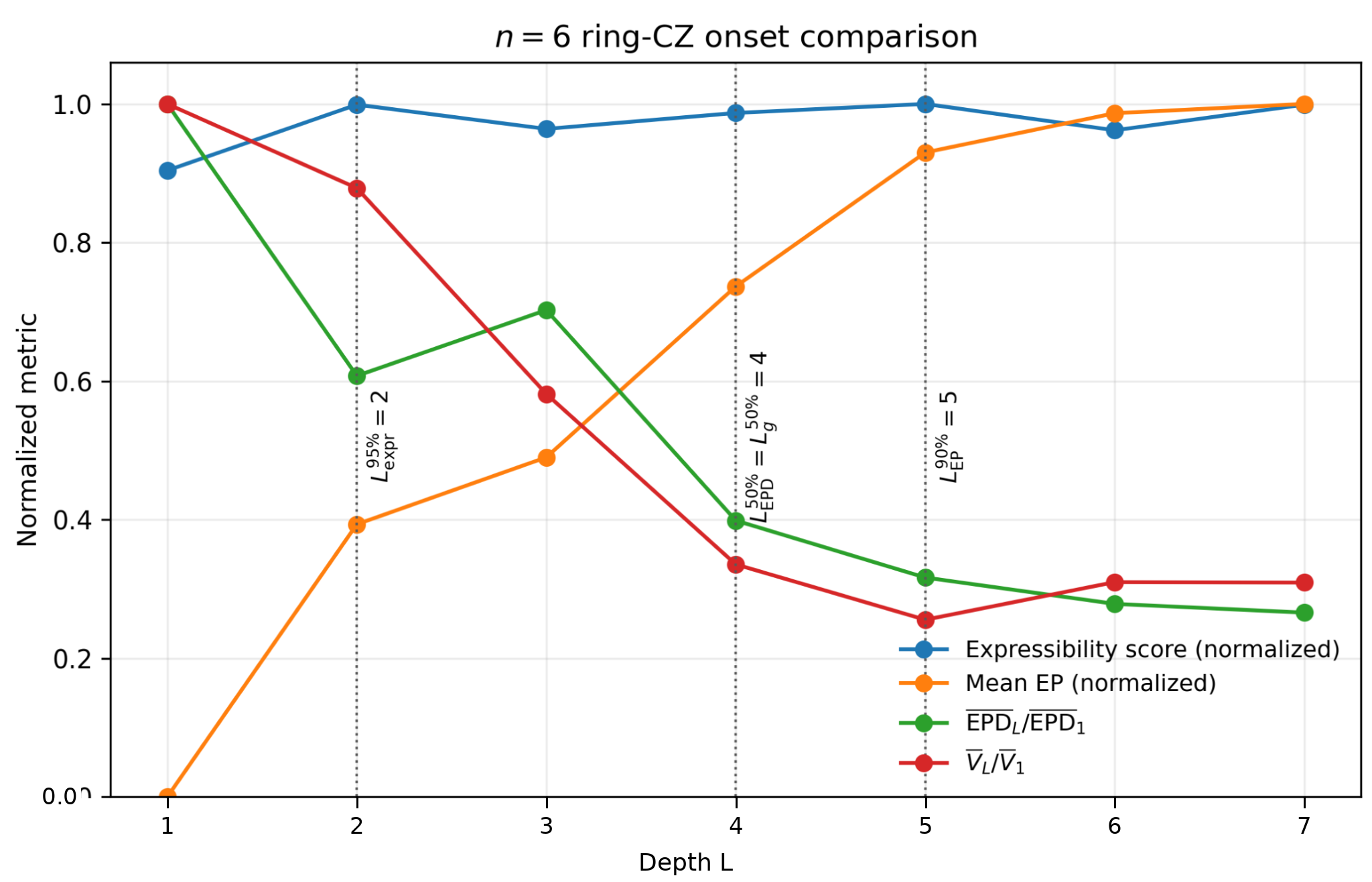}
\caption{Normalized onset comparison for the $n=6$ ring-CZ ansatz. The coverage proxy reaches a near-ceiling value at $L=2$, while EPD and the representative gradient variance both pass their half-collapse thresholds at $L=4$ in this finite-sample run.}
\label{fig:sm_onset}
\end{figure}

Fig.~\ref{fig:sm_onset} restates the same information in threshold language. For the finite-sample $n=6$ scan, the ordering is
\begin{eqnarray}
L^{95\%}_{\rm expr}=2,
\qquad
L^{90\%}_{\rm EP}=5,
\qquad
L^{50\%}_{\rm EPD}=4,
\qquad
L^{50\%}_{g}=4.
\label{eq:sm_thresholds_expanded}
\end{eqnarray}
The meaning of these four depths is summarized in Table~\ref{tab:sm_onset_meaning}. The thresholds are not universal constants; they are a compact way to compare crossover events in one reproducible finite-size scan. The robust observation is the ordering: coverage can become large before EPD and the representative gradient variance have collapsed.

\begin{table}[h]
\centering
\caption{Interpretation of the onset depths in Eq.~(\ref{eq:sm_thresholds_expanded}). The values are finite-sample diagnostics for the $n=6$ ring-CZ scan, not universal constants.}
\label{tab:sm_onset_meaning}
\begin{tabular}{ccl}
\toprule
quantity & depth & \parbox[t]{0.75\textwidth}{interpretation} \\
\midrule
$L^{95\%}_{\rm expr}$ & 2 & \parbox[t]{0.75\textwidth}{the Sim-style coverage proxy has reached $95\%$ of its maximum value in the depth scan; coverage is already near its finite-sample ceiling.} \\
$L^{90\%}_{\rm EP}$ & 5 & \parbox[t]{0.75\textwidth}{the normalized mean fixed-circuit EP reaches $90\%$ of its scan range; average entangling strength continues to grow after the coverage proxy has saturated.} \\
$L^{50\%}_{\rm EPD}$ & 4 & \parbox[t]{0.75\textwidth}{EPD has fallen to one half of its $L=1$ value; product-input response variability has significantly narrowed.} \\
$L^{50\%}_{g}$ & 4 & \parbox[t]{0.62\textwidth}{the representative gradient variance has fallen to one half of its $L=1$ value; this finite-difference probe aligns with the EPD-collapse depth in this scan.} \\
\bottomrule
\end{tabular}
\end{table}

The interval between early coverage saturation and later variability collapse is the numerical prototype of the design window emphasized in the main manuscript. In this window, the ansatz is no longer trivially shallow, but it has not yet reached the low-EPD regime associated with BP-like homogenization.

Fig.~\ref{fig:sm_hist} shows histograms of $Q(\hU_L\ket{\psi})$ over Haar-random product inputs for $n=6$ ring-CZ circuits. This plot is the most direct way to understand what EP and EPD measure. EP tracks the horizontal drift of the distribution; EPD tracks its width.

\begin{figure}[t]
\centering
\begin{tabular}{c}
\includegraphics[width=0.98\textwidth]{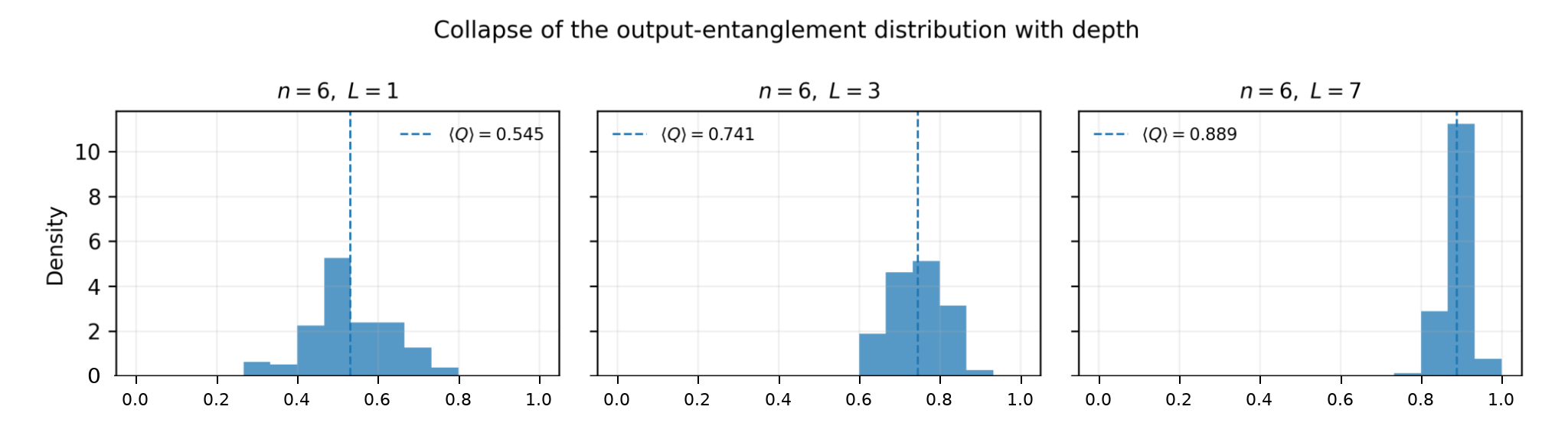}\\[-0.35em]
{\small $Q(\hU_L(\bm{\theta})\ket{\psi})$}
\end{tabular}
\caption{Collapse of the output-entanglement distribution with depth. Histograms of $Q(\hU_L\ket{\psi})$ are shown for representative $n=6$ ring-CZ circuits at $L=1,3,7$. Increasing depth shifts the distribution to larger $Q$ and then narrows it.}
\label{fig:sm_hist}
\end{figure}

The three histograms should be read as follows. At $L=1$, the circuit is shallow. The distribution is relatively broad because different product inputs experience noticeably different entangling responses. At the intermediate depth, the distribution has moved to the right while still retaining a visible width; this is the desired qualitative signature of high coverage without full homogenization. At larger depth, the histogram is concentrated near a high-entanglement value. The mean is high, but the width is small. This last regime is precisely why EP alone is insufficient: a high-EP circuit can be either useful and still input-sensitive, or high-EP and already homogenized.

A small EPD therefore has two different meanings depending on coverage. A shallow circuit can have small EPD simply because it barely entangles any product input. A deep circuit can have small EPD because it entangles almost all product inputs in the same way. For this reason, EPD should always be read together with a coverage coordinate.

\subsection{Ansatz trajectory scan}

The depth scan above studies one ansatz family. The route scan generalizes the same idea to five families: ring-CZ HEA, line-CNOT HEA, brickwork RXX, QAOA-style ring-ZZ, and ring-CRX HEA. Fig.~\ref{fig:sm_l1_circuits} shows the representative $n=5$ and $L=1$ circuit blocks used in the route scan. The drawings use five wires only as a schematic stencil; for larger $n$, the same local pattern is extended along the line or ring.

\begin{figure}[t]
\centering
\includegraphics[width=0.84\textwidth]{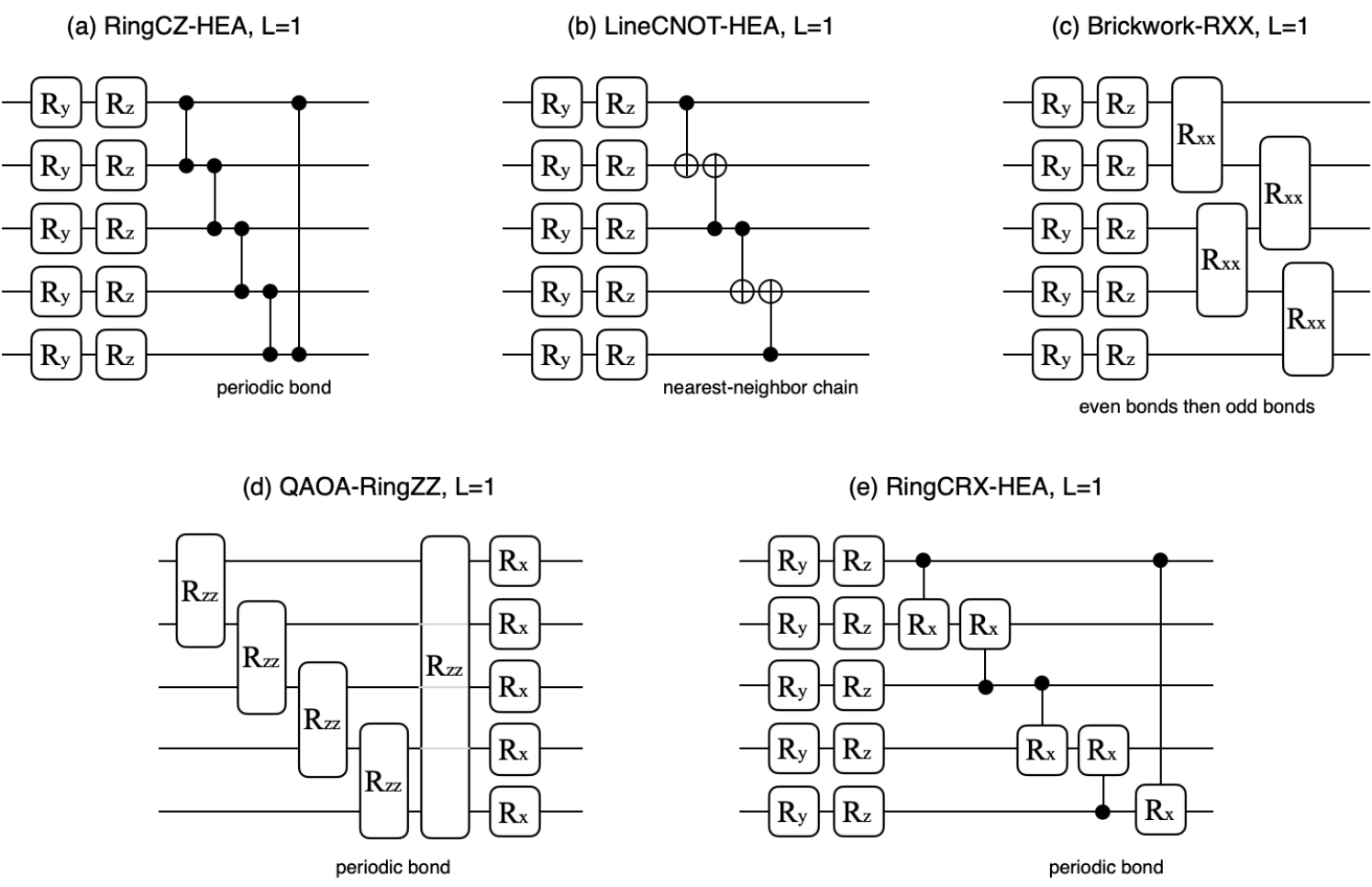}
\caption{Representative $L=1$ circuit blocks for the five ansatz families used in the route scan~\cite{SM-Sim2019expressibility,SM-Kandala2017HardwareEfficient,SM-Farhi2014QAOA}. Top row: (a) ring-CZ HEA, (b) line-CNOT HEA, and (c) brickwork RXX. Bottom row: (d) QAOA-style ring-ZZ and (e) ring-CRX HEA. The figure is drawn for five qubits for readability; the simulations extend the same connectivity pattern to $n=5,10,15,20$.}
\label{fig:sm_l1_circuits}
\end{figure}

The route-scan simulations use the following elementary gates:
\begin{eqnarray}
\hR_X(\theta)=e^{-i\theta \hX/2},
\quad
\hR_Y(\theta)=e^{-i\theta \hY/2},
\quad
\hR_Z(\theta)=e^{-i\theta \hZ/2},
\quad
\hR_{XX}(\theta)=e^{-i\theta \hX\otimes \hX/2},
\quad
\hR_{ZZ}(\theta)=e^{-i\theta \hZ\otimes \hZ/2},
\end{eqnarray}
with standard $CZ$, $CNOT$, and controlled-$\hR_X$ gates. The five ansatz families are:

\begin{itemize}
\item \textit{RingCZ-HEA}: each layer applies $\hR_Y$ and $\hR_Z$ rotations to every qubit, followed by a periodic ring of $CZ$ gates.
\item \textit{LineCNOT-HEA}: each layer applies $\hR_Y$ and $\hR_Z$ rotations to every qubit, followed by a nearest-neighbor CNOT chain. The control direction is alternated in the implementation across layers to avoid a fixed orientation bias.
\item \textit{Brickwork-RXX}: each layer applies $\hR_Y$ and $\hR_Z$ rotations to every qubit, followed by $\hR_{XX}$ gates on alternating even and odd nearest-neighbor bonds.
\item \textit{QAOA-RingZZ}: each layer applies $\hR_{ZZ}(\gamma_\ell)$ gates on the ring followed by single-qubit $\hR_X(\beta_\ell)$ mixers.
\item \textit{RingCRX-HEA}: each layer applies $\hR_Y$ and $\hR_Z$ rotations to every qubit, followed by controlled-$\hR_X$ gates on a ring, with the directed pattern alternated across layers.
\end{itemize}

These families were chosen to produce visibly different route geometries. Line-like entanglers can move rapidly along the coverage direction but may also lose EPD quickly. Brickwork interactions can preserve variability longer while accumulating coverage more slowly. Ring-based and controlled-rotation hardware-efficient families often create intermediate-depth windows. QAOA-style ring-ZZ uses a more constrained parameter structure and can enter the sweet region at a different depth scale.

The representative gradient-variance probe uses the global projector cost $\hO_0=\ketbra{0^n}{0^n}$ and the first trainable rotation parameter of each ansatz family. For hardware-efficient ansatze this is the first-layer single-qubit $\hR_Y$ parameter on qubit $0$; for the QAOA-style ansatz it is the first mixer angle. In the reproduction package this diagnostic is evaluated with a small-budget central finite difference. It is included only to compare the EPD collapse with a concrete cost-dependent gradient scale.

The simulation budgets are intentionally modest and are stored in the metadata file. They are summarized in Table~\ref{tab:sm_budget} below. Increasing these budgets improves Monte Carlo precision without altering the pipeline.
\begin{table}[h]
\centering
\setlength{\tabcolsep}{20pt}
\begin{tabular}{cccc}
\toprule
qubits $n$ & $N_\theta$ & $N_\psi$ & $N_{\rm targ}$ \\
\midrule
5  & 16 & 40 & 5 \\
10 & 8  & 24 & 6 \\
15 & 4  & 12 & 4 \\
20 & 2  & 4  & 2 \\
\bottomrule
\end{tabular}
\caption{Monte Carlo budgets used in the route-scan script. $N_\theta$ is the number of parameter samples at each ansatz-depth-width point, $N_\psi$ is the number of product-input probes per circuit sample, and $N_{\rm targ}$ is the number of one-qubit reductions sampled when target-qubit subsampling is used for larger systems.}
\label{tab:sm_budget}
\end{table}

At each $n\in\{5,10,15,20\}$ and $L=1,\ldots,6$, the code estimates
\begin{eqnarray}
\overline{\EP}_{A,L}^{(n)}=\mathbb E_{\bm{\theta}\sim\nu_{A,L}^{(n)}}[\EP(\hU_{A,L}^{(n)}(\bm{\theta}))],
\qquad
\overline{\EPD}_{A,L}^{(n)}=\mathbb E_{\bm{\theta}\sim\nu_{A,L}^{(n)}}[\EPD(\hU_{A,L}^{(n)}(\bm{\theta}))],
\label{eq:sm_family_averages}
\end{eqnarray}
where $A$ labels the ansatz family and $\nu_{A,L}^{(n)}$ is the ansatz-, depth-, and size-dependent product measure over trainable angles. Eq.~(\ref{eq:sm_family_averages}) again encodes the fixed-circuit ordering: for each sampled $\bm{\theta}$, EP and EPD are first estimated over product-Haar inputs, and only then averaged over the parameter ensemble.

\subsubsection{Algorithmic pipeline} 

The numerical pipeline is deliberately written in fixed-circuit order. For each sampled parameter vector, the code first estimates the product-input distribution of $Q(\hU_L(\bm{\theta})\ket{\psi})$ and records its mean and standard deviation. Only after this fixed-circuit step are the resulting EP and EPD values averaged over parameter samples. This order is essential: averaging $Q$ over parameter samples from a fixed reference input would reproduce the Sim-style entangling-capability convention rather than the fixed-circuit EP/EPD convention used here.

The numerical steps are summarized in Algorithm~\ref{alg:sm_pipeline}.
\begin{algorithm}[H]
\caption{Fixed-circuit EP/EPD route scan and gradient diagnostic}\label{alg:sm_pipeline}
\begin{algorithmic}[1]
\Require Ansatz family $A$, qubit-size $n$, depth $L$, parameter measure $\nu_{A,L}^{(n)}$, product-Haar measure $\mu_{\mathrm{prod}}$, budgets $N_\theta,N_\psi$, and finite-difference step $\delta$.
\For{$r=1,\ldots,N_\theta$}
    \State Sample $\bm{\theta}^{(r)}\sim\nu_{A,L}^{(n)}$ and instantiate the fixed circuit $\hU_{A,L}^{(n)}(\bm{\theta}^{(r)})$.
    \For{$q=1,\ldots,N_\psi$}
        \State Sample a product-Haar input $\ket{\psi^{(q)}}\sim\mu_{\mathrm{prod}}$.
        \State Propagate $\ket{\psi^{(q)}}$ through the fixed circuit.
        \State Compute $Q(\hU_{A,L}^{(n)}(\bm{\theta}^{(r)})\ket{\psi^{(q)}})$ from one-qubit reduced purities.
        \State For the optional gradient, evaluate the global-projector cost at $\bm{\theta}^{(r)}_{+i}$ and $\bm{\theta}^{(r)}_{-i}$, where only the designated parameter is shifted by $\pm\delta$.
    \EndFor
    \State Estimate the fixed-circuit $\EP$ and $\EPD$ from the $N_\psi$ values of $Q$.
    \State Estimate the input-induced variance of the central finite-difference gradients, when the probe is used.
\EndFor
\State Average the fixed-circuit estimates over $r=1,\ldots,N_\theta$ to obtain $\overline{\EP}_{A,L}^{(n)}$ and $\overline{\EPD}_{A,L}^{(n)}$.
\State Normalize the route cloud to obtain $x_{A,L}^{(n)}$ and $y_{A,L}^{(n)}$.
\State Apply the maximum-separation rule to compute $x_c^{(n)}$ and $y_c^{(n)}$, and assign each point to $\mathcal U^{(n)}$, $\mathcal S^{(n)}$, or $\mathcal P^{(n)}$.
\end{algorithmic}
\end{algorithm}

For the probe, we use
\begin{eqnarray}
g_i(\psi)\simeq
\frac{\abs{\bra{0^n}\hU(\bm{\theta}_{+i})\ket{\psi}}^2-
\abs{\bra{0^n}\hU(\bm{\theta}_{-i})\ket{\psi}}^2}{2\delta},
\qquad \delta=10^{-6},
\label{eq:sm_fd_gradient}
\end{eqnarray}
where $\bm{\theta}_{\pm i}$ denotes the same parameter vector as $\bm{\theta}$ except that the $i$th component is shifted to $\theta_i\pm\delta$.
The analytic results in the main manuscript do not depend on this numerical probe. It is a cost-dependent consistency check showing that low-EPD regions often coincide with smaller representative gradient scales, as expected from the local two-copy criterion, but it is not a substitute for a task-specific trainability analysis.


\subsubsection{Operational definition of the sweet spot} 

For each qubit number $n$, we convert the family-averaged descriptors into two panel-normalized coordinates,
\begin{eqnarray}
x_{A,L}^{(n)}=\frac{\overline{\EP}_{A,L}^{(n)}}{Q_{\Haar}^{(n)}},
\qquad
 y_{A,L}^{(n)}=\frac{\overline{\EPD}_{A,L}^{(n)}}{\max_{A,L}\overline{\EPD}_{A,L}^{(n)}},
\label{eq:sm_route_xy}
\end{eqnarray}
where
\begin{eqnarray}
Q_{\Haar}^{(n)}=\frac{2^n-2}{2^n+1}
\label{eq:sm_qhaar_route}
\end{eqnarray}
is the Haar mean of the Meyer-Wallach entanglement. The horizontal coordinate $x$ is the coverage dial: it measures how close the family-averaged mean entanglement is to the Haar benchmark at the same $n$. The vertical coordinate $y$ is the variability dial: it measures how much input-dependent output-entanglement spread remains relative to the largest EPD value observed in the same $n$ panel.

The shaded regions in the main Fig.~1 are fixed by the route data, not by hand-tuned constants. For a finite set $Z=\{z_m\}_{m=1}^M\subset[0,1]$, let $\mathcal T(Z)$ be the midpoint cuts between adjacent sorted distinct values. For $\tau\in\mathcal T(Z)$ define
\begin{eqnarray}
W_Z(\tau)=\sum_{z_m<\tau}(z_m-\bar z_<)^2+
\sum_{z_m\ge\tau}(z_m-\bar z_\ge)^2,
\label{eq:sm_within_cluster}
\end{eqnarray}
where $\bar z_<$ and $\bar z_\ge$ are the means of the two groups. We set
\begin{eqnarray}
\tau_*(Z)=\arg\min_{\tau\in\mathcal T(Z)} W_Z(\tau).
\label{eq:sm_tau_star}
\end{eqnarray}
Equivalently, $\tau_*$ is the one-dimensional Otsu/Fisher split that maximizes the separation between the two classes while keeping each class internally compact~\cite{SM-Otsu1979Threshold,SM-Fisher1936Taxonomic}. This is a useful choice because it introduces no external threshold, percentile, or gradient information.

The coverage boundary is the split of all coverage coordinates in a panel,
\begin{eqnarray}
x_c^{(n)}=\tau_*\!\left(\{x_{A,L}^{(n)}\}_{A,L}\right).
\label{eq:sm_xc_rule}
\end{eqnarray}
The variability boundary is then computed only inside the high-coverage subset,
\begin{eqnarray}
y_c^{(n)}=\tau_*\!\left(\{y_{A,L}^{(n)}:x_{A,L}^{(n)}\ge x_c^{(n)}\}_{A,L}\right).
\label{eq:sm_yc_rule}
\end{eqnarray}
The conditional second step is essential. Low variability at low coverage can simply mean that the circuit is too weak to produce a broad entanglement distribution. It is not evidence of BP-like homogenization. The EPD-collapse cut is therefore learned only from points that already passed the coverage cut.

The three regions used in the route scan are
\begin{eqnarray}
\mathcal U^{(n)} &=& \{(A,L):x_{A,L}^{(n)}<x_c^{(n)}\},\nonumber\\
\mathcal S^{(n)} &=& \{(A,L):x_{A,L}^{(n)}\ge x_c^{(n)},\;y_{A,L}^{(n)}\ge y_c^{(n)}\},\nonumber\\
\mathcal P^{(n)} &=& \{(A,L):x_{A,L}^{(n)}\ge x_c^{(n)},\;y_{A,L}^{(n)}<y_c^{(n)}\}.
\label{eq:sm_regions}
\end{eqnarray}
Here, $\mathcal U$ is the underexpressive region: the coverage dial has not crossed the data-driven coverage cut. High variability here is not enough, because the family has not yet reached substantial coverage. $\mathcal S$ is the sweet spot: the coverage dial is high and the variability dial remains open. This is the desired high-coverage/nonhomogenized design window. $\mathcal P$ is the near-plateau or homogenized high-coverage region

The reproduced route scan gives
\begin{eqnarray}
(x_c^{(5)},y_c^{(5)})&=&(0.731,0.472),\nonumber\\
(x_c^{(10)},y_c^{(10)})&=&(0.700,0.352),\nonumber\\
(x_c^{(15)},y_c^{(15)})&=&(0.698,0.321),\nonumber\\
(x_c^{(20)},y_c^{(20)})&=&(0.712,0.263).
\label{eq:sm_threshold_values}
\end{eqnarray}
The absolute threshold values should not be interpreted as universal physical constants. What is meaningful is the route geometry: a point in $\mathcal S^{(n)}$ has already moved far along the coverage direction but still retains a broad output-entanglement distribution over product inputs, whereas a point in $\mathcal P^{(n)}$ has comparable coverage but has lost that variability.

For a compact one-number ranking inside each panel, we use the descriptive sweet score
\begin{eqnarray}
s_{A,L}^{(n)}=x_{A,L}^{(n)}y_{A,L}^{(n)}.
\label{eq:sm_sweet_score}
\end{eqnarray}
The score rewards points that are simultaneously far to the right and high in the two-dial plane. It is not a new objective function, and it is not used as a theorem hypothesis. In particular, an underexpressive point with very large $y$ can still have a sizable $s$. Therefore Eq.~(\ref{eq:sm_sweet_score}) must be read together with the region classification in Eq.~(\ref{eq:sm_regions}).

\begin{table}[h]
\centering
\caption{Meaning of the factors in the sweet score $s_{A,L}^{(n)}=x_{A,L}^{(n)}y_{A,L}^{(n)}$.}
\label{tab:sm_sweet_score_meaning}
\begin{tabular}{ccl}
\toprule
symbol & role & \parbox[t]{0.75\textwidth}{interpretation} \\
\midrule
$x_{A,L}^{(n)}$ & coverage factor & \parbox[t]{0.75\textwidth}{large when the family-averaged EP is close to the Haar entanglement benchmark. It prevents a shallow low-coverage circuit from being selected only because it has high variability.} \\
$y_{A,L}^{(n)}$ & variability factor & \parbox[t]{0.75\textwidth}{large when the fixed-circuit product-input response remains broad. It distinguishes high-coverage/nonhomogenized circuits from high-coverage/homogenized circuits.} \\
$s_{A,L}^{(n)}$ & descriptive rank & \parbox[t]{0.75\textwidth}{useful for ranking candidate depths within a panel, but not a universal trainability metric. It should be interpreted only after checking whether the point lies in $\mathcal U$, $\mathcal S$, or $\mathcal P$.} \\
\bottomrule
\end{tabular}
\end{table}

\subsubsection{Detailed reading of the route geometry} 

Fig.~1 in the main manuscript plots the trajectories of the five ansatz families in the $(x,y)$ plane as $L$ increases from $1$ to $6$. The route geometry carries the physical information. Moving to the right means that the family has gained mean entangling coverage. Moving downward means that the product-input response has homogenized. A useful route moves sufficiently far right before it moves too far down.

For $n=5$, two types of behavior are already visible. LineCNOT-HEA starts in the sweet region and then rapidly moves toward high coverage with reduced variability. RingCZ-HEA and RingCRX-HEA pass through an intermediate region before descending toward lower EPD. Brickwork-RXX retains appreciable variability but remains mostly under the coverage cut. QAOA-RingZZ reaches the sweet region only at later depth. Thus even for small $n$, EP alone would merge late high-coverage points with genuinely useful sweet-spot points, while EPD separates them. For $n=10$, the intermediate-depth windows become clearer. RingCZ-HEA and RingCRX-HEA develop visible sweet-spot portions before falling into the near-plateau sector. LineCNOT-HEA reaches high coverage early and then loses variability quickly. Brickwork-RXX remains mostly underexpressive, despite retaining width in the output-entanglement distribution. QAOA-RingZZ approaches the useful region more slowly. The lesson is that faster coverage is not automatically better: a route that arrives at high $x$ only after $y$ has collapsed is less attractive than one that enters $\mathcal S^{(n)}$. For $n=15$, the distinction between coverage and homogenization is sharper. LineCNOT-HEA and the deepest RingCZ-HEA points lie at high $x$ but low $y$, illustrating why a one-dimensional EP-vs-BP reading is incomplete. RingCZ-HEA and RingCRX-HEA still show middle-depth sweet windows, whereas QAOA-RingZZ reaches a useful region at a later depth. Brickwork-RXX again demonstrates the opposite failure mode: noncollapsed variability without enough coverage. For $n=20$, the newly added panel confirms that the two-dial interpretation is not restricted to the smaller scans. The value of $y_c^{(20)}$ is lower than in the smaller panels because it is a panel-wise split of the normalized diagnostic cloud, not a universal BP threshold. RingCZ-HEA retains a middle-depth high-coverage/nonhomogenized window before descending. LineCNOT-HEA largely occupies the high-coverage/low-variability side, with only a narrow sweet-spot passage. Brickwork-RXX and QAOA-RingZZ reach the useful region later. RingCRX-HEA also illustrates the importance of depth scheduling: too shallow is underexpressive, whereas too deep tends toward homogenization.

The route-scan conclusion is therefore structural rather than family-specific. The different entanglers and connectivities trace different paths through the two-dial plane. The design task is to choose circuit blocks, connectivities, and depth schedules that move far enough along the coverage dial while keeping the variability dial open.

Table~\ref{tab:sm_threshold_counts} below also explains why the route plot should be read panel by panel. The coverage split is relatively stable near $0.7$ for this diagnostic cloud, while the variability split decreases with $n$. This does not mean that $y_c$ is a system-size-independent trainability threshold. It means that the normalized EPD cloud changes with system size and must be separated within each panel.
\begin{table}[h]
\centering
\setlength{\tabcolsep}{20pt}
\begin{tabular}{cccccc}
\toprule
$n$ & $x_c^{(n)}$ & $y_c^{(n)}$ & $\abs{\mathcal U^{(n)}}$ & $\abs{\mathcal S^{(n)}}$ & $\abs{\mathcal P^{(n)}}$ \\
\midrule
5  & 0.731 & 0.472 & 14 & 6 & 10 \\
10 & 0.700 & 0.352 & 15 & 8 & 7 \\
15 & 0.698 & 0.321 & 15 & 6 & 9 \\
20 & 0.712 & 0.263 & 15 & 7 & 8 \\
\bottomrule
\end{tabular}
\caption{Panel-wise thresholds and region counts for the reproduced route scan. The counts are over the $5$ ansatz families and $6$ depths in each $n$ panel.}
\label{tab:sm_threshold_counts}
\end{table}

Table~\ref{tab:sm_sweet_summary} below lists the data-driven thresholds and, for each ansatz family and qubit number, the depth that maximizes $s_{A,L}^{(n)}$. The additional region column is important: the largest sweet score for a family is not automatically a sweet-spot point if the coverage cut is not passed.
\begin{table}[h]
\centering
\setlength{\tabcolsep}{13pt}
\begin{adjustbox}{max width=0.98\textwidth}
\begin{tabular}{ccclccccc}
\toprule
$n$ & $x_c^{(n)}$ & $y_c^{(n)}$ & ansatz & best $L$ & $x_*$ & $y_*$ & $s_*$ & region \\
\midrule
5 & 0.731 & 0.472 & RingCZ-HEA     & 1 & 0.659 & 1.000 & 0.659 & $\mathcal U$ \\
5 & 0.731 & 0.472 & LineCNOT-HEA   & 1 & 0.760 & 0.766 & 0.582 & $\mathcal S$ \\
5 & 0.731 & 0.472 & Brickwork-RXX  & 6 & 0.660 & 0.650 & 0.429 & $\mathcal U$ \\
5 & 0.731 & 0.472 & QAOA-RingZZ    & 6 & 0.797 & 0.525 & 0.418 & $\mathcal S$ \\
5 & 0.731 & 0.472 & RingCRX-HEA    & 2 & 0.694 & 0.649 & 0.450 & $\mathcal U$ \\
\midrule
10 & 0.700 & 0.352 & RingCZ-HEA     & 1 & 0.596 & 1.000 & 0.596 & $\mathcal U$ \\
10 & 0.700 & 0.352 & LineCNOT-HEA   & 1 & 0.739 & 0.713 & 0.527 & $\mathcal S$ \\
10 & 0.700 & 0.352 & Brickwork-RXX  & 5 & 0.614 & 0.781 & 0.480 & $\mathcal U$ \\
10 & 0.700 & 0.352 & QAOA-RingZZ    & 4 & 0.640 & 0.659 & 0.422 & $\mathcal U$ \\
10 & 0.700 & 0.352 & RingCRX-HEA    & 4 & 0.814 & 0.535 & 0.435 & $\mathcal S$ \\
\midrule
15 & 0.698 & 0.321 & RingCZ-HEA     & 1 & 0.629 & 0.944 & 0.594 & $\mathcal U$ \\
15 & 0.698 & 0.321 & LineCNOT-HEA   & 1 & 0.806 & 0.537 & 0.433 & $\mathcal S$ \\
15 & 0.698 & 0.321 & Brickwork-RXX  & 3 & 0.565 & 0.729 & 0.412 & $\mathcal U$ \\
15 & 0.698 & 0.321 & QAOA-RingZZ    & 1 & 0.550 & 1.000 & 0.550 & $\mathcal U$ \\
15 & 0.698 & 0.321 & RingCRX-HEA    & 3 & 0.762 & 0.479 & 0.365 & $\mathcal S$ \\
\midrule
20 & 0.712 & 0.263 & RingCZ-HEA     & 1 & 0.666 & 0.585 & 0.389 & $\mathcal U$ \\
20 & 0.712 & 0.263 & LineCNOT-HEA   & 3 & 0.855 & 0.328 & 0.281 & $\mathcal S$ \\
20 & 0.712 & 0.263 & Brickwork-RXX  & 4 & 0.571 & 0.546 & 0.312 & $\mathcal U$ \\
20 & 0.712 & 0.263 & QAOA-RingZZ    & 5 & 0.665 & 0.524 & 0.349 & $\mathcal U$ \\
20 & 0.712 & 0.263 & RingCRX-HEA    & 1 & 0.535 & 1.000 & 0.535 & $\mathcal U$ \\
\bottomrule
\end{tabular}
\end{adjustbox}
\caption{Route-scan thresholds and best sweet-score summary. The thresholds $x_c^{(n)}$ and $y_c^{(n)}$ are computed by Eqs.~(\ref{eq:sm_xc_rule}) and (\ref{eq:sm_yc_rule}). $x_*$ and $y_*$ are the coordinates at the depth maximizing $s_{A,L}^{(n)}=x_{A,L}^{(n)}y_{A,L}^{(n)}$, and $s_*$ is the corresponding sweet score. The final column reports the region of that best-score point.}
\label{tab:sm_sweet_summary}
\end{table}

The score table reinforces the same caution. In several rows, the maximum score occurs in $\mathcal U^{(n)}$ because $y$ is large while $x$ is still below the coverage cut. These points should not be read as final sweet-spot candidates. The practical screening rule is: first identify whether the route reaches $\mathcal S^{(n)}$, then use the score and task-specific gradient checks to rank nearby depths.

\section{Task-level validation on QML benchmarks}\label{subsec:task_validation}

The route scan in the previous Sec.~V is task-independent. It uses only fixed-circuit EP/EPD descriptors and therefore cannot, by itself, certify high accuracy for every possible dataset, data encoding, optimizer, readout, and loss. A natural question is therefore whether the high-coverage/nonhomogenized region selected by the two-dial scan is truly useful in an actual learning loop. In this section, we address this question by adding a small descriptor-blind QML validation. The purpose is not to turn the sweet score into a universal accuracy theorem. Rather, it is to test whether the scan identifies ansatz-depth points at which task-specific training is most worthwhile.

The validation uses the $n=5$ route cloud from Fig.~1 of the main text. All $30$ ansatz-depth points, namely the five ansatz families at $L=1,\ldots,6$, are trained under the same supervised-learning protocol. The EP/EPD coordinates, region assignments, and sweet scores are fixed before training and are not adjusted using the task labels. Thus, the comparison is descriptor-blind: the learning task only tests the candidates selected by the scan. The two summary figures for this validation are included in the main text; the present section keeps the full protocol, tables, and interpretation.

\subsection{Benchmark construction and training protocol}

We use a simple teacher-student binary classification benchmark. For each input $\bm{x}=(x_1,x_2)\in[-1,1]^2$, the data embedding prepares
\begin{eqnarray}
\ket{\eta(\bm{x})}=\bigotimes_{j=1}^{5}\hat R_z\!\left(b_j(\bm{x})\right)\hat R_y\!\left(a_j(\bm{x})\right)\ket{0},
\label{eq:sm_task_embedding}
\end{eqnarray}
where the angles $a_j,b_j$ are fixed linear combinations of $x_1$ and $x_2$. The same embedding is used for every ansatz-depth point. We then generate labels using two fixed sweet-region teacher circuits, $\hat{T}_1=({\rm RingCZ\text{-}HEA},L=3)$ and $\hat{T}_2=({\rm RingCRX\text{-}HEA},L=3)$, both of which lie in the green region of the $n=5$ route scan. For teacher $\hat{T}$, we define
\begin{eqnarray}
z_T(\bm{x})=\bra{\eta(\bm{x})}\hat U_T^\dagger \hat Z_0 \hat U_T\ket{\eta(\bm{x})}.
\label{eq:sm_teacher_score}
\end{eqnarray}
The binary label is obtained by thresholding $z_T(\bm{x})$ at the median value on the training set, so that the two classes are approximately balanced:
\begin{eqnarray}
y_T(\bm{x})=\begin{cases}
+1, & z_T(\bm{x})\ge m_T,\\
-1, & z_T(\bm{x})<m_T,
\end{cases}
\qquad
m_T={\rm median}_{\bm{x}\in\mathcal D_{\rm train}}z_T(\bm{x}).
\label{eq:sm_teacher_label}
\end{eqnarray}
This teacher-student construction is useful because it probes an actual QML learning loop while keeping the target function in the representational regime that the two-dial scan is designed to find: high coverage is useful, but the teacher itself is not taken from the low-EPD near-plateau region.

For a student ansatz $(A,L)$, the trainable prediction is
\begin{eqnarray}
f_{A,L}(\bm{x};\bm{\theta})=\bra{\eta(\bm{x})}\hat U_{A,L}(\bm{\theta})^\dagger \hat Z_0 \hat U_{A,L}(\bm{\theta})\ket{\eta(\bm{x})},
\label{eq:sm_student_readout}
\end{eqnarray}
and the loss is the mean squared error
\begin{eqnarray}
\mathcal L_{A,L}(\bm{\theta})=\frac{1}{|\mathcal D_{\rm train}|}\sum_{(\bm{x},y)\in\mathcal D_{\rm train}}\left(f_{A,L}(\bm{x};\bm{\theta})-y\right)^2.
\label{eq:sm_task_loss}
\end{eqnarray}
All circuits are initialized from the same type of random angle distribution used in the route scan, namely independent uniform angles in $[0,2\pi)$. We train with full-batch Adam-SPSA for $50$ iterations. SPSA is used only to keep the benchmark architecture-neutral: each update requires two loss evaluations independent of the number of parameters. For each ansatz-depth point and each teacher task, we use two independent optimizer seeds. The training set has $40$ samples and the test set has $96$ samples.

The validation metrics are listed in Table~\ref{tab:task_metric_definitions}. A run is called successful when its final test accuracy is at least $75\%$. The initial gradient scale is the per-parameter squared norm of the first SPSA gradient estimate,
\begin{eqnarray}
G_0=\frac{1}{p_{A,L}}\left\|\widehat{\nabla}_{\rm SPSA}\mathcal L_{A,L}(\bm{\theta}_0)\right\|_2^2,
\label{eq:sm_initial_spsa_scale}
\end{eqnarray}
which should be read as an optimizer-level proxy for the available initial training signal, not as the exact local two-copy variance of Sec.~III.
\begin{table}[h]
\centering
\setlength{\tabcolsep}{13pt}
\begin{tabular}{lcl}
\toprule
metric & symbol & interpretation \\
\midrule
final test accuracy & $A_{\rm test}$ & classification accuracy after the fixed training budget \\
final training loss & $\mathcal L_{\rm final}$ & residual optimization error on the training set \\
success probability & $P_{\rm succ}$ & fraction of runs with $A_{\rm test}\ge0.75$ \\
first success iteration & $T_{\rm succ}$ & first logged iteration at which the success threshold is reached \\
initial gradient scale & $G_0$ & per-parameter squared norm of the first SPSA gradient estimate \\
\bottomrule
\end{tabular}
\caption{Task-level validation metrics. The benchmark is not used to define EP, EPD, or the sweet spot; it is used only after the descriptor scan to test whether the selected ansatz-depth regions behave as useful QML candidates.}
\label{tab:task_metric_definitions}
\end{table}

\subsection{Results: sweet region as a useful screening window, not a guarantee}

Table~\ref{tab:task_region_summary} groups the task-level results by the three EP/EPD regions. The most relevant comparison for the two-dial claim is the high-coverage comparison between the green and red regions. Both have passed the coverage cut, but they differ in whether the variability dial remains open. In this comparison, the sweet region has higher median test accuracy, lower median final loss, larger success probability, and a much faster first-success time than the near-plateau region. In particular, the median first-success iteration is $5$ in the sweet region but $20$ in the near-plateau region. This supports the intended interpretation of the green region as the depth window in which high-coverage training is most worthwhile.
\begin{table}[h]
\centering
\setlength{\tabcolsep}{13pt}
\begin{tabular}{lcccccc}
\toprule
region & runs & median $A_{\rm test}$ & median $\mathcal L_{\rm final}$ & $P_{\rm succ}$ & median $T_{\rm succ}$ & median $G_0$ \\
\midrule
underexpressive & 56 & 0.682 & 0.793 & 0.357 & 12.5 & $3.51\times10^{-1}$ \\
sweet spot & 24 & 0.656 & 0.804 & 0.333 & 5.0 & $1.72\times10^{-1}$ \\
near plateau & 40 & 0.609 & 0.865 & 0.200 & 20.0 & $2.91\times10^{-1}$ \\
\bottomrule
\end{tabular}
\caption{Region-wise task-level validation for the $n=5$ route cloud. The table reports medians over all ansatz-depth points, the two teacher tasks, and two optimizer seeds. The success threshold is $A_{\rm test}\ge0.75$.}
\label{tab:task_region_summary}
\end{table}

The underexpressive region deserves a careful interpretation. Some gray-region circuits perform competitively on these small $n=5$ teacher tasks. This is not a contradiction of the two-dial framework. It means that a small finite task may not require the full coverage that the descriptor scan regards as desirable. The gray region can therefore contain trainable circuits for simple tasks, but it is not the high-coverage design window targeted by the main Letter. The operational question addressed by the sweet spot is different: once high coverage has been reached, should one continue increasing depth into the homogenized red sector? The answer from Table~\ref{tab:task_region_summary} is no; the high-coverage/noncollapsed region gives better task-level behavior than the high-coverage/low-EPD region under the same training protocol.

The two task-level validation plots have been moved to the main text. The EP/EPD scatter plot shown there gives the same result directly in the route plane: the color indicates the empirical success probability of each ansatz-depth point, and the high-coverage/low-EPD sector contains several low-success points. Successful points are more often found away from the collapsed lower-right corner. This is exactly the design information lost by projecting the route cloud onto the coverage coordinate alone.

The region-wise diagnostic plot in the main text isolates three aggregate quantities. The first panel shows that sweet-region success is slightly below the gray region but above the red near-plateau region. The second panel is the most diagnostic for trainability: among successful runs, the sweet region reaches the success threshold earlier than both gray and red regions. The third panel reports the initial SPSA gradient scale. These quantities are not expected to be monotonic functions of EPD, because the exact gradient object is local and cost-dependent. Nevertheless, the red sector does not provide an advantage despite having the largest coverage; this is the task-level counterpart of the main two-dial message.

We also compare simple screening rules in Table~\ref{tab:task_screening_rules} below. The rule based on the product $s=xy$ selects points that balance coverage and variability and gives a larger success probability than selecting by coverage alone. Selecting by variability alone can also perform well in this small task, but it includes low-coverage gray points; this is why the main text defines the sweet spot by applying the variability cut only after the coverage cut. Depth alone and parameter count alone are not reliable substitutes for the two-dial geometry.
\begin{table}[h]
\centering
\setlength{\tabcolsep}{15pt}
\begin{tabular}{lccc}
\toprule
screening rule & median $A_{\rm test}$ & mean $P_{\rm succ}$ & median $G_0$ \\
\midrule
sweet score $s=xy$ & 0.648 & 0.375 & $3.51\times10^{-1}$ \\
coverage $x$ only & 0.570 & 0.125 & $4.11\times10^{-1}$ \\
variability $y$ only & 0.703 & 0.417 & $5.26\times10^{-1}$ \\
deepest $L$ only & 0.698 & 0.333 & $2.72\times10^{-1}$ \\
fewest parameters & 0.589 & 0.208 & $3.22\times10^{-1}$ \\
all points & 0.641 & 0.300 & $3.46\times10^{-1}$ \\
\bottomrule
\end{tabular}
\caption{Comparison of simple descriptor-blind screening rules. Each rule selects six ansatz-depth points from the $n=5$ route cloud before any task training. The table then reports the median test accuracy, mean success probability, and median initial gradient scale of the selected points.}
\label{tab:task_screening_rules}
\end{table}

Finally, Table~\ref{tab:task_predictor_corr} reports Spearman rank correlations between scalar predictors and task success. The coverage coordinate alone is negatively correlated with success in this benchmark, because many of the highest-coverage points lie in the low-EPD near-plateau sector. By contrast, the variability coordinate and the sweet score have positive correlations with success. These numbers should not be overinterpreted as universal constants. Their role is to verify the central logic: once one moves beyond very small tasks, EP-only or coverage-only screening is incomplete, and retaining noncollapsed variability is useful information for ansatz selection.
\begin{table}[h]
\centering
\setlength{\tabcolsep}{15pt}
\begin{tabular}{lc}
\toprule
predictor & Spearman $\rho$ with $P_{\rm succ}$ \\
\midrule
coverage dial $x$ & -0.251 \\
variability dial $y$ & 0.374 \\
sweet score $s=xy$ & 0.364 \\
mean EP & -0.251 \\
EPD & 0.374 \\
depth $L$ & -0.061 \\
parameter count & -0.102 \\
\bottomrule
\end{tabular}
\caption{Spearman rank correlation between descriptor-level predictors and success probability in the task-level validation. Positive values indicate that larger predictor values tend to rank ansatz-depth points with higher success probability.}
\label{tab:task_predictor_corr}
\end{table}

The validation therefore answers the practical question in a deliberately limited sense. A high sweet score is not a universal guarantee of high QML accuracy, and the present paper does not claim such a no-free-lunch reversal. Instead, the EP/EPD scan is a pre-training filter. It identifies the regime in which coverage has been gained without obvious homogenization, and the task-level check shows that this regime is a better high-coverage target than the low-EPD near-plateau sector.

\section{Interpretive remarks and limitations of the diagnostic}

The central theoretical hierarchy is
\begin{eqnarray}
\mathrm{coverage} &:& \Tr[\hK_Q\hM^{(2)}_{\hU}],\nonumber\\
\mathrm{variability} &:& \Tr[\hT_Q\hM^{(4)}_{\hU}]-\EP^2,\nonumber\\
\mathrm{trainability} &:& \Tr[(\hB_i\otimes \hB_i)\hM^{(2)}_{i,\LC}].
\label{eq:sm_final_hierarchy}
\end{eqnarray}
The first term is global and mean-like, the second is global but fluctuation-sensitive, and the third is local and parameter-resolved. This is why the Letter does not claim that EPD is the barren-plateau invariant. Rather, EPD is a compact fourth-moment diagnostic of the same homogenization process that, when seen locally by a parameter light cone, suppresses gradient variance.

This distinction also clarifies how to interpret small EPD. A shallow underexpressive circuit may have small EPD because it generates little entanglement at all. A deep randomizing circuit may have small EPD because it generates almost the same high entanglement for most product inputs. The diagnostic is therefore most informative in combination with a coverage coordinate: high normalized EP together with low EPD indicates high-coverage/low-variability homogenization. The route plot in the Letter is designed precisely to show this joint behavior.

Several extensions are natural. One may replace the Meyer-Wallach functional by other entanglement or correlation functionals; one may replace the Haar-product input distribution by a data-conditioned distribution; and one may adapt the moment formulas to noisy channels. These generalizations change the operational meaning of the descriptors but not the main lesson: mean coverage, fluctuation width, and local trainability are distinct objects and should not be collapsed into a single scalar trade-off.

\end{document}